\documentclass[11pt]{article}
\usepackage{CJK}
\usepackage{latexsym,bm}
\usepackage{xypic}
\usepackage{amsmath}
\usepackage{wasysym}
\usepackage{indentfirst}
\usepackage{amssymb}
\usepackage{dsfont}
\usepackage{amsthm}
\begin{document}
\newcommand{\fr}[2]{\frac{\;#1\;}{\;#2\;}}
\newtheorem{theorem}{Theorem}[section]
\newtheorem{lemma}{Lemma}[section]
\newtheorem{proposition}{Proposition}[section]
\newtheorem{corollary}{Corollary}[section]
\newtheorem{conjecture}{Conjecture}[section]
\newtheorem{remark}{Remark}[section]
\newtheorem{definition}{Definition}[section]
\newtheorem{example}{Example}[section]
\newtheorem{notation}{Notation}[section]
\numberwithin{equation}{section}
\newcommand{\Aut}{\mathrm{Aut}\,}
\newcommand{\CSupp}{\mathrm{CSupp}\,}
\newcommand{\Supp}{\mathrm{Supp}\,}
\newcommand{\rank}{\mathrm{rank}\,}
\newcommand{\col}{\mathrm{col}\,}
\newcommand{\len}{\mathrm{len}\,}
\newcommand{\leftlen}{\mathrm{leftlen}\,}
\newcommand{\rightlen}{\mathrm{rightlen}\,}
\newcommand{\length}{\mathrm{length}\,}
\newcommand{\wt}{\mathrm{wt}\,}
\newcommand{\diff}{\mathrm{diff}\,}
\newcommand{\lcm}{\mathrm{lcm}\,}
\newcommand{\dom}{\mathrm{dom}\,}
\newcommand{\SUPP}{\mathrm{SUPP}\,}
\newcommand{\supp}{\mathrm{supp}\,}
\newcommand{\End}{\mathrm{End}\,}
\newcommand{\Hom}{\mathrm{Hom}\,}
\newcommand{\ran}{\mathrm{ran}\,}
\newcommand{\Mat}{\mathrm{Mat}\,}
\newcommand{\rk}{\mathrm{rk}\,}
\newcommand{\rs}{\mathrm{rs}\,}
\newcommand{\piv}{\mathrm{piv}\,}
\newcommand{\perm}{\mathrm{perm}\,}
\newcommand{\inv}{\mathrm{inv}\,}
\newcommand{\orb}{\mathrm{orb}\,}
\newcommand{\id}{\mathrm{id}\,}
\newcommand{\soc}{\mathrm{soc}\,}

\title{Reflexivity of Partitions Induced by Weighted Poset Metric and Combinatorial Metric$^\ast$}
\author{Yang Xu$^1$ \,\,\,\,\,\, Haibin Kan$^2$\,\,\,\,\,\,Guangyue Han$^3$}

\maketitle

\renewcommand{\thefootnote}{\fnsymbol{footnote}}

\footnotetext{\hspace*{-12mm} \begin{tabular}{@{}r@{}p{13.4cm}@{}}
$^\ast$ & A preliminary version of this paper has been presented in IEEE International Symposium on Information Theory (ISIT) 2022.\\
$^1$ & Shanghai Key Laboratory of Intelligent Information Processing, School of Computer Science, Fudan University,
Shanghai 200433, China.\\
&Department of Mathematics, Faculty of Science, The University of Hong Kong, Pokfulam Road, Hong Kong, China. {E-mail:12110180008@fudan.edu.cn} \\
$^2$ & Shanghai Key Laboratory of Intelligent Information Processing, School of Computer Science, Fudan University,
Shanghai 200433, China.\\
&Shanghai Engineering Research Center of Blockchain, Shanghai 200433, China.\\
&Yiwu Research Institute of Fudan University, Yiwu City, Zhejiang 322000, China. {E-mail:hbkan@fudan.edu.cn} \\
$^3$ & Department of Mathematics, Faculty of Science, The University of Hong Kong, Pokfulam Road, Hong Kong, China. {E-mail:ghan@hku.hk} \\
\end{tabular}}

\vskip 3mm

{\hspace*{-6mm}{\bf Abstract}---Let $\mathbf{H}$ be the Cartesian product of a family of finite abelian groups. Via a polynomial approach, we give sufficient conditions for a partition of $\mathbf{H}$ induced by weighted poset metric to be reflexive, which also become necessary for some special cases. Moreover, by examining the roots of the Krawtchouk polynomials, we establish non-reflexive partitions of $\mathbf{H}$ induced by combinatorial metric. When $\mathbf{H}$ is a vector space over a finite field $\mathbb{F}$, we consider the property of admitting MacWilliams identity (PAMI) and the MacWilliams extension property (MEP) for partitions of $\mathbf{H}$. With some invariance assumptions, we show that two partitions of $\mathbf{H}$ admit MacWilliams identity if and only if they are mutually dual and reflexive, and any partition of $\mathbf{H}$ satisfying the MEP is in fact an orbit partition induced by some subgroup of $\Aut_{\mathbb{F}}(\mathbf{H})$, which is necessarily reflexive. As an application of the aforementioned results, we establish partitions of $\mathbf{H}$ induced by combinatorial metric that do not satisfy the MEP, which further enable us to provide counter-examples to a conjecture proposed by Pinheiro, Machado and Firer in \cite{39}.
}\!

\section{Introduction}
For a set $E$, a \textit{partition} of $E$ is a collection of nonempty disjoint subsets of $E$ whose union is $E$. MacWilliams identities based on partitions of finite abelian groups have been established by Zinoviev and Ericson in \cite{45}, by Gluesing-Luerssen in \cite{18}, and by Gluesing-Luerssen and Ravagnani in \cite{21} (c.f. Delsarte \cite{10} and Ravagnani \cite{40} for MacWilliams identities based on association schemes and numerical weights, respectively). The above mentioned work provides a general and unified framework for recovering known or deriving new MacWilliams identities in more explicit forms.

The notion of \textit{reflexive partition} is introduced by Gluesing-Luerssen in \cite{18}. More specifically, reflexive partitions are ones which coincide with their bi-duals, and alternatively, they can be characterized in terms of association schemes (see \cite{10,11,46}). Reflexive partitions arise naturally from various weights and metrics in coding theory such as poset metric (see \cite{5,22,35}), rank metric (see \cite{17,21}) and homogeneous weight (see \cite{19}). We refer the reader to \cite{18,19,21,40,44,45,46} for more results and examples. The MacWilliams identity based on a reflexive partition is invertible, and the inverse is essentially the MacWilliams identity based on the dual partition. So, in this sense, reflexive partitions provide a symmetric situation and form the most appealing case (see [18, Section 2]).

In this paper, we study partitions induced by weighted poset metric and combinatorial metric, and examine when such partitions are reflexive or non-reflexive. Moreover, we study the relations among reflexivity, the property of admitting MacWilliams identity (PAMI) and the MacWilliams extension property (MEP), three widely explored properties in coding theory.

The notion of \textit{weighted poset metric} has been introduced by Hyun, Kim and Park in \cite{23}, where the authors have classified all the weighted posets and directed graphs that admit the extended Hamming code $\widetilde{\mathcal{H}}_{3}$ to be a $2$-perfect code, and relevant results for more general $\widetilde{\mathcal{H}}_{k}$, $k\geqslant3$ have also been established. A weighted poset metric is determined by a poset and a weight function, both defined on the coordinate set (see \cite{23} or Section 2.2 for more details). Weighted poset metric boils down to poset metric (see \cite{5,22,35}) if the weight function is identically $1$, and to weighted Hamming metric (see \cite{2}) if the poset is an anti-chain. As has been stated in [23, Section 1], weighted poset metric can be viewed as an algebraic version of the directed graph metric introduced by Etzion, Firer and Machado in \cite{13}. More recently in \cite{29}, Machado and Firer have proposed and studied the labelled-poset-block metric, which, in our terminology, is also a weighted poset metric (see Section 2.2). Weighted poset metric can be useful to model some specific kind of channels for which the error probability depends on a codeword position, i.e., the distribution of errors is nonuniform, and can also be useful to perform bitwise or messagewise unequal error protection (see the abstract of \cite{2} and [13, Section I, Paragraph 6]).

The notion of \textit{combinatorial metric} has been introduced by Gabidulin in \cite{15,16}. A combinatorial metric is determined by a covering of the coordinate set (see Section 2.3 for more details). If the covering consists of singletons, then combinatorial metric boils down to Hamming metric. Several subclasses of combinatorial metric have been studied in the literature, such as the block metric (see \cite{14}), the $b$-burst metric (see \cite{4}) and the translational metric (see \cite{33}). In \cite{3}, Bossert and Sidorenko have derived a Singleton-type bound for combinatorial metric. In \cite{39}, Pinheiro, Machado and Firer have studied the PAMI, the group of isometries and the MEP for combinatorial metric. They have also proposed a conjecture in [39, Section 5] on the MEP, which we will disprove in this paper. Our approach towards the conjecture is based on non-reflexive partitions induced by combinatorial metric along with the relation between reflexivity and the MEP.

The PAMI has been first introduced by Kim and Oh in \cite{24}, where the authors have proven that being hierarchical is a necessary and sufficient condition for a poset to admit MacWilliams identity. The original property has since been extended and generalized to poset-block metric by Pinheiro and Firer in \cite{38}, to combinatorial metric by Pinheiro, Machado and Firer in \cite{39}, to directed graph metric by Etzion, Machado and Firer in \cite{13}, and to labeled-poset-block metric by Machado and Firer in \cite{29}. In \cite{8}, Choi, Hyun, Kim and Oh have proposed and studied MacWilliams-type equivalence relations, which, roughly speaking, are defined as equivalence relations which admit MacWilliams identities, where such relations are defined on the ideal lattice of a given poset on the coordinate set.

In 1962, MacWilliams proved in \cite{30} that a Hamming weight preserving map between two linear codes can be extended to the whole ambient space (see \cite{6} for a different proof). Such a property, henceforth referred to as the MEP, has since been extended, generalized and discussed extensively in the literature: with respect to other weights and metrics; with respect to codes over ring and module alphabet; and with respect to partitions of finite modules; see, among many others, \cite{1}, \cite{12}, \cite{13}, \cite{17}, \cite{20}, \cite{21}, \cite{28}, \cite{29}, \cite{39}, \cite{42} and \cite{43}.

The remainder of the paper is organized as follows. In Section 2, we present some definitions, notations and basic facts on partitions of finite abelian groups, weighted poset metric and combinatorial metric. In particular, the conjecture proposed in \cite{39} is stated in Section 2.3 as Conjecture 2.1. In Section 3, for a partition induced by weighted poset metric, we give sufficient conditions for two codewords to belong to the same member of its dual partition, and give a sufficient condition for its reflexivity. By relating each codeword with a polynomial, we show that such sufficient conditions are also necessary if the poset is hierarchical and the weight function is integer-valued. In Section 4, we consider partitions induced by combinatorial metric. After giving a sufficient but not necessary condition for reflexivity in Section 4.1, we turn to a particular subclass of partitions which is related to Conjecture 2.1 in Section 4.2. Adopting the polynomial approach given in Section 3, we characterize the dual partitions of such partitions in terms of the classical Krawtchouk polynomials. Then, using the properties of Krawtchouk polynomials, especially those of their roots, we give a number of sufficient conditions for such partitions to be non-reflexive. In Section 5, we consider partitions of a finite vector space over a finite field $\mathbb{F}$, and study the relations among reflexivity, the PAMI and the MEP for $\mathbb{F}$-invariant partitions. More precisely, we show that for such partitions, the MEP is stronger than reflexivity, and reflexivity is equivalent to the PAMI. Finally, combining these results with non-reflexive partitions established in Section 4.2, we give several classes of counter-examples to Conjecture 2.1.

\section{Preliminaries}
Throughout the remainder of the paper, we let $\mathbb{Z}$, $\mathbb{Z}^{+}$, $\mathbb{R}$, $\mathbb{R}^{+}$ and $\mathbb{C}$ denote the set of all the integers, positive integers, real numbers, positive real numbers and complex numbers, respectively. Furthermore, we let $\mathbb{N}=\mathbb{Z}^{+}\cup\{0\}$, $\mathbb{C}^{\ast}=\mathbb{C}-\{0\}$. For any $a,b\in\mathbb{Z}$, we use $[a,b]$ to denote the set of all the integers between $a$ and $b$, i.e., $[a,b]=\{i\in\mathbb{Z}\mid a\leqslant i\leqslant b\}$.

In addition, consider a finite set $E$. A \textit{covering} of $E$ is a collection of subsets of $E$ whose union is $E$, and hence a partition of $E$ is exactly a covering of $E$ whose members are nonempty and disjoint. Consider a partition $\Gamma$ of $E$. For any $u,v\in E$, we write $u\sim_{\Gamma}v$ if $u$ and $v$ belong to the same member of $E$. For any $D\subseteq E$, the \textit{$\Gamma$-distribution} of $D$ is defined as the sequence $(|D\cap B|\mid B\in\Gamma)$, and for any $D,L\subseteq E$, we write $D\approx_{\Gamma}L$ if $D$ and $L$ have the same $\Gamma$-distribution. For two partitions $\Gamma,\Psi$ of $E$, we say that \textit{$\Gamma$ is finer than $\Psi$}, if for any $u,v\in E$, $u\sim_{\Gamma}v$ implies that $u\sim_{\Psi}v$. One can verify that $\Gamma$ is finer than $\Psi$ if and only if any member of $\Gamma$ is contained in some member of $\Psi$.

\subsection{Partitions of finite abelian groups}
Let $G$ and $H$ be finite abelian groups, and let $f:G\times H\longrightarrow \mathbb{C}^{*}$ be a \textit{pairing}, i.e., for any $a,c\in G$ and $b,d\in H$, it holds that $f(ac,b)=f(a,b)f(c,b)$, $f(a,bd)=f(a,b)f(a,d)$ (see [34, Definition 11.7]). For any \textit{additive codes} (\textit{i.e.}, subgroups) $C\leqslant G$ and $D\leqslant H$, define the codes $C^{\ddagger}\leqslant H$ and $^{\ddagger}D\leqslant G$ as
\begin{equation}C^{\ddagger}\triangleq\{b\in H\mid \text{$f(a,b)=1$ for all $a\in C$}\},\end{equation}
\begin{equation}{^{\ddagger}D}\triangleq\{a\in G\mid \text{$f(a,b)=1$ for all $b\in D$}\}.\end{equation}
We further assume that the pairing $f$ is \textit{non-degenerate}, i.e., $G^{\ddagger}=\{1_{H}\}$, ${^{\ddagger}H}=\{1_{G}\}$ (see [34, Definition 11.7]). Note that the non-degenerate condition implies that $G\cong H$ as groups, and conversely, $G\cong H$ implies the existence of such a non-degenerate pairing (see [34, Lemma 11.8]).

For a partition $\Gamma$ of $H$, \textit{the left dual partition of $\Gamma$ with respect to $f$} is the partition $\textbf{\textit{l}}(\Gamma)$ of $G$ such that for any $a,c\in G$, $a\sim_{\textbf{\textit{l}}(\Gamma)}c$ if and only if
\begin{equation}\text{$\sum_{b\in B}f(a,b)=\sum_{b\in B}f(c,b)$ for all $B\in\Gamma$}.\end{equation}
For a partition $\Lambda$ of $G$, \textit{the right dual partition of $\Lambda$ with respect to $f$} is the partition $\textbf{\textit{r}}(\Lambda)$ of $H$ such that for any $b,d\in H$, $b\sim_{\textbf{\textit{r}}(\Lambda)}d$ if and only if
\begin{equation}\text{$\sum_{a\in A}f(a,b)=\sum_{a\in A}f(a,d)$ for all $A\in\Lambda$}.\end{equation}
Equations (2.3) and (2.4) are closely related to the notion of \textit{dual partition} and \textit{bi-dual partition} proposed in \cite{18}. Indeed, for any $a\in G$, define $\tau_{a}:H\longrightarrow\mathbb{C}^{*}$ as $\tau_{a}(b)=f(a,b)$. By [34, Lemma 11.8], we have $\tau\in\Aut(G,\Hom(H,\mathbb{C}^{*}))$. Consequently, for a partition $\Gamma$ of $H$, we have $\{\tau[A]\mid A\in\textbf{\textit{l}}(\Gamma)\}=\widehat{\Gamma}$, $\textbf{\textit{r}}(\textbf{\textit{l}}(\Gamma))=\widehat{\widehat{\Gamma}}$, which are exactly the dual partition and the bi-dual partition of $\Gamma$ proposed in [18, Definition 2.1] (c.f. [1, Proposition 4.4], [7, Section 2.1], [19, Proposition 2.4]).

The following definition follows [18, Definition 2.1] and [46, Definition 2].

\setlength{\parindent}{0em}
\begin{definition}
{Let $\Gamma$ be a partition of $H$, and let $\Lambda$ be a partition of $G$. If $\Lambda$ is finer than $\textbf{\textit{l}}(\Gamma)$, then the left generalized Krawtchouk matrix of $(\Lambda,\Gamma)$ with respect to $f$ is defined as $\rho:\Lambda\times\Gamma\longrightarrow \mathbb{C}$, where for any $(A,B)\in\Lambda\times\Gamma$, $\rho(A,B)=\sum_{b\in B}f(a,b)$ for any chosen $a\in A$. If $\Gamma$ is finer than $\textbf{\textit{r}}(\Lambda)$, then the right generalized Krawtchouk matrix of $(\Lambda,\Gamma)$ with respect to $f$ is defined as $\varepsilon:\Lambda\times\Gamma\longrightarrow \mathbb{C}$, where for any $(A,B)\in\Lambda\times\Gamma$, $\varepsilon(A,B)=\sum_{c\in A}f(c,d)$ for any chosen $d\in B$. $(\Lambda,\Gamma)$ is said to be mutually dual with respect to $f$ if both $\Lambda$ is finer than $\textbf{\textit{l}}(\Gamma)$ and $\Gamma$ is finer than $\textbf{\textit{r}}(\Lambda)$. Finally, $\Gamma$ is said to be reflexive if $\Gamma=\widehat{\widehat{\Gamma}}$.
}
\end{definition}

\setlength{\parindent}{2em}
The following lemma is a consequence of [18, Theorem 2.4].

\setlength{\parindent}{0em}
\begin{lemma}
{Let $\Gamma$ be a partition of $H$, and let $\Lambda$ be a partition of $G$. Then, we have $\{1_{G}\}\in\textbf{\textit{l}}(\Gamma)$, $|\Gamma|\leqslant|\textbf{\textit{l}}(\Gamma)|$, $\textbf{\textit{r}}(\textbf{\textit{l}}(\Gamma))$ is finer than $\Gamma$. Moreover, it holds true that $\Gamma$ is reflexive $\Longleftrightarrow \Gamma=\textbf{\textit{r}}(\textbf{\textit{l}}(\Gamma))\Longleftrightarrow  |\Gamma|=|\textbf{\textit{l}}(\Gamma)|$. In addition, the following four statements are equivalent to each other:

{\bf{(1)}}\,\,$(\Lambda,\Gamma)$ is mutually dual with respect to $f$;

{\bf{(2)}}\,\,$\Gamma$ is Fourier-reflexive and $\Lambda=\textbf{\textit{l}}(\Gamma)$;

{\bf{(3)}}\,\,$\Lambda$ is Fourier-reflexive and $\Gamma=\textbf{\textit{r}}(\Lambda)$;

{\bf{(4)}}\,\,$|\Lambda|\leqslant|\Gamma|$ and $\Lambda$ is finer than $\textbf{\textit{l}}(\Gamma)$.
}
\end{lemma}

\setlength{\parindent}{2em}
For partitions $\Gamma$ of $H$ and $\Lambda$ of $G$ such that $\Lambda$ is finer than $\textbf{\textit{l}}(\Gamma)$, let $\rho:\Lambda\times\Gamma\longrightarrow\mathbb{C}$ be the left generalized Krawtchouk matrix of $(\Lambda,\Gamma)$ with respect to $f$. It has been proven in [18, Theorem 2.7] that for an additive code $C\leqslant G$, the $\Lambda$-distribution of $C$ determines the $\Gamma$-distribution of $C^{\ddagger}$ via the following MacWilliams identity
\begin{equation}\forall~B\in\Gamma:|C||{C^{\ddagger}}\cap B|=\sum_{A\in \Lambda}|C\cap A|\cdot\rho(A,B).\end{equation}
Consequently, we have
\begin{equation}\text{$C_1\thickapprox_{\Lambda}C_2\Longrightarrow {C_1}^{\ddagger}\thickapprox_{\Gamma}{C_2}^{\ddagger}$ for any $C_1,C_2\leqslant G$}.\end{equation}

The following theorem can be viewed as a partial converse of the fact that ``$\Lambda$ is finer than $\textbf{\textit{l}}(\Gamma)$ implies Equation (2.6)''.

\setlength{\parindent}{0em}
\begin{theorem}
{Let $\textbf{S}$ be a collection of non-identity subgroups of $G$ with the same cardinality, and let $\Delta$ be a partition of $G$ such that $\{1_{G}\}\in\Delta$, and for any $A\in\Delta$ with $A\neq\{1_{G}\}$, there exists $C\in\textbf{S}$ such that $C-\{1_{G}\}\subseteq A$. Let $\Gamma$ be a partition of $H$ such that $\Delta$ is finer than $\textbf{\textit{l}}(\Gamma)$, and let $\Lambda$ be a partition of $G$ such that $\{1_{G}\}\in\Lambda$, $\Delta$ is finer than $\Lambda$. Further assume that for any $C,M\in\textbf{S}$, we have $C\approx_{\Lambda}M\Longrightarrow C^{\ddagger}\approx_{\Gamma}M^{\ddagger}$. Then, $\Lambda$ is finer than $\textbf{\textit{l}}(\Gamma)$.
}
\end{theorem}

\begin{proof}
Letting $W\in\Lambda$ and $a,c\in W$, we will show that $a\sim_{\textbf{\textit{l}}(\Gamma)}c$, which immediately yields the desired result. Let $U,V\in\Delta$ such that $a\in U$, $c\in V$. Since $\Delta$ is finer than $\Lambda$, we have $U\subseteq W$, $V\subseteq W$. If $W=\{1_{G}\}$, then $a=c=1_{G}$, as desired. Therefore in the following, we assume that $W\neq\{1_{G}\}$. By $\{1_{G}\}\in\Lambda$, we have $1_{G}\not\in W$, which further implies that $U\neq\{1_{G}\}$, $V\neq\{1_{G}\}$. Hence we can choose $C,M\in\textbf{S}$ such that $C-\{1_{G}\}\subseteq U$, $M-\{1_{G}\}\subseteq V$. We note that $|C|=|M|\geqslant2$. It is straightforward to verify that $C\cap W=C-\{1_{G}\}$, $M\cap W=M-\{1_{G}\}$, $|C\cap \{1_{G}\}|=|M\cap \{1_{G}\}|=1$, and for any $A\in\Lambda$ such that $A\neq W$, $A\neq\{1_{G}\}$, it holds that $C\cap A=M\cap A=\emptyset$. Therefore we have $C\approx_{\Lambda}M$, and hence $C^{\ddagger}\approx_{\Gamma}M^{\ddagger}$. Consider an arbitrary $B\in\Gamma$. Since $\Delta$ is finer than $\textbf{\textit{l}}(\Gamma)$, $\{1_{G}\}\in\Delta$, $U\in\Delta$, $C\cap U=C-\{1_{G}\}$, $a\in U$, $V\in\Delta$, $M\cap V=M-\{1_{G}\}$, $c\in V$, applying (2.5) to $C\leqslant G$ and $M\leqslant G$, respectively, we deduce that
$$|C||{C^{\ddagger}}\cap B|=|B|+(|C|-1)\left(\sum_{b\in B}f(a,b)\right),$$
$$|M||{M^{\ddagger}}\cap B|=|B|+(|M|-1)\left(\sum_{b\in B}f(c,b)\right).$$
Since $C^{\ddagger}\approx_{\Gamma}M^{\ddagger}$, $B\in\Gamma$, we have $|{C^{\ddagger}}\cap B|=|{M^{\ddagger}}\cap B|$, which, along with $|C|=|M|\geqslant2$, immediately implies that $\sum_{b\in B}f(a,b)=\sum_{b\in B}f(c,b)$. It then follows from the arbitrariness of $B$ that $a\sim_{\textbf{\textit{l}}(\Gamma)}c$, as desired.
\end{proof}

\begin{remark}
Theorem 2.1 is largely inspired by [8, Corollary 3.2 and Theorem 3.3], and will be used in Section 5 to establish the equivalence between reflexivity and the PAMI.
\end{remark}

\subsection{Weighted poset metric}
\setlength{\parindent}{2em}
Throughout this subsection, we let $\Omega$ be a nonempty finite set, and let $\mathbf{P}=(\Omega,\preccurlyeq_{\mathbf{P}})$ be a poset. A subset $B\subseteq \Omega$ is said to be an \textit{ideal} of $\mathbf{P}$ if for any $v\in B$ and $u\in \Omega$, $u\preccurlyeq_{\mathbf{P}}v$ implies that $u\in B$. The set of all the ideals of $\mathbf{P}$ is denoted by $\mathcal{I}(\mathbf{P})$. For $B\subseteq \Omega$, we let $\max_{\mathbf{P}}(B)$ and $\min_{\mathbf{P}}(B)$ denote the set of all the maximal (\textit{resp.}, minimal) element of $B$, and let $\langle B\rangle_{\mathbf{P}}$ denote the ideal $\{u\in \Omega\mid \exists~v\in B~s.t.~u\preccurlyeq_{\mathbf{P}}v\}$. In addition, $B$ is said to be a \textit{chain} in $\mathbf{P}$ if for any $u,v\in B$, either $u\preccurlyeq_{\mathbf{P}}v$ or $v\preccurlyeq_{\mathbf{P}}u$ holds, and $B$ is said to be an \textit{anti-chain} in $\mathbf{P}$ if for any $u,v\in B$, $u\preccurlyeq_{\mathbf{P}}v$ implies that $u=v$. For any $u\in \Omega$, we let $\len_{\mathbf{P}}(u)$ denote the largest cardinality of a chain in $\mathbf{P}$ containing $u$ as its greatest element. The set of all the order automorphisms of $\mathbf{P}$ will be denoted by $\Aut(\mathbf{P})$. The \textit{dual poset} of $\mathbf{P}$ is defined as $\mathbf{\overline{P}}=(\Omega,\preccurlyeq_{\mathbf{\overline{P}}})$, where
$$\text{$u\preccurlyeq_{\mathbf{\overline{P}}} v\Longleftrightarrow v\preccurlyeq_{\mathbf{P}}u$ for all $(u,v)\in \Omega\times \Omega$}.$$

The following definition will be used frequently in our discussion.

\setlength{\parindent}{0em}
\begin{definition}
{{\bf{(1)}}\,\,$\mathbf{P}$ is said to be hierarchical if for any $u,v\in \Omega$ such that $\len_{\mathbf{P}}(u)+1\leqslant\len_{\mathbf{P}}(v)$, it holds that $u\preccurlyeq_{\mathbf{P}}v$.

{\bf{(2)}}\,\,For $\omega:\Omega\longrightarrow\mathbb{R}^{+}$, we say that $(\mathbf{P},\omega)$ satisfies the unique decomposition property (UDP) if for any $I,J\in\mathcal{I}(\mathbf{P})$ such that $\sum_{i\in I}\omega(i)=\sum_{j\in J}\omega(j)$, there exists $\lambda\in\Aut(\mathbf{P})$ such that $J=\lambda[I]$ and $\omega(i)=\omega(\lambda(i))$ for all $i\in \Omega$.
}
\end{definition}

\setlength{\parindent}{2em}
The following lemma is an immediate consequence of Definition 2.2 and the fact that $\mathcal{I}(\mathbf{\overline{P}})=\{\Omega-I\mid I\in\mathcal{I}(\mathbf{P})\}$ (see [22, Lemma 1.2]).

\setlength{\parindent}{0em}
\begin{lemma}
{Let $\omega:\Omega\longrightarrow\mathbb{R}^{+}$. Then, $(\mathbf{P},\omega)$ satisfies the UDP if and only if $(\mathbf{\overline{P}},\omega)$ satisfies the UDP.
}
\end{lemma}

\setlength{\parindent}{2em}
Now we let $\left(H_{i}\mid i\in \Omega\right)$ be a family of finite abelian groups, and let $\mathbf{H}\triangleq\prod_{i\in \Omega}H_{i}$. For any \textit{codeword} $\beta\in \mathbf{H}$, we let $\supp(\beta)$ denote the set
\begin{equation}\supp(\beta)=\{i\in \Omega\mid \beta_{i}\neq1_{H_{i}}\}.\end{equation}
Consider $\omega:\Omega\longrightarrow\mathbb{R}^{+}$. For any $\beta\in\mathbf{H}$, the $(\mathbf{P},\omega)$-weight of $\beta$ is defined as \begin{equation}\wt_{(\mathbf{P},\omega)}(\beta)\triangleq\sum_{i\in\langle\supp(\beta)\rangle_{\mathbf{P}}}\omega(i).\end{equation}
It has been proven in [23, Lemma I.2] that $d_{(\mathbf{P},\omega)}:\mathbf{H}\times\mathbf{H}\longrightarrow \mathbb{R}$ defined as $d_{(\mathbf{P},\omega)}(\alpha,\beta)=\wt_{(\mathbf{P},\omega)}(\alpha^{-1}\beta)$ induces a metric on $\mathbf{H}$, which is henceforth referred to as a weighted poset metric. We note that if $\mathbf{H}$ is a vector space over a finite field and the weight function is integer-valued, then the weighted poset metric coincides with the labeled-poset-block metric proposed in [29, Section III].

Finally, we introduce partitions induced by weighted poset metric.

\setlength{\parindent}{0em}
\begin{notation}
For $\omega:\Omega\longrightarrow\mathbb{R}^{+}$, we let $\mathcal{Q}(\mathbf{H},\mathbf{P},\omega)$ denote the partition of $\mathbf{H}$ such that for any $\beta,\theta\in\mathbf{H}$, $\beta\sim_{\mathcal{Q}(\mathbf{H},\mathbf{P},\omega)}\theta\Longleftrightarrow\wt_{(\mathbf{P},\omega)}(\beta)=\wt_{(\mathbf{P},\omega)}(\theta)$.
\end{notation}

\subsection{Combinatorial metric}
\setlength{\parindent}{2em}
Throughout this subsection, we let $\Omega$ be a nonempty finite set. For a covering $T$ of $\Omega$, define $\omega_{T}:2^{\Omega}\longrightarrow\mathbb{N}$ as
\begin{equation}\mbox{$\forall~A\subseteq\Omega:\omega_{T}(A)=\min\{|S|\mid S\subseteq T,A\subseteq\bigcup_{I\in S}I\}$}.\end{equation}
For any $r\in\mathbb{N}$, we let $\mathcal{P}(r,\Omega)$ denote the set of all the subsets of $\Omega$ with cardinality $r$, i.e.,
\begin{equation}\mathcal{P}(r,\Omega)=\{A\subseteq\Omega\mid |A|=r\}.\end{equation}

We collect some basic facts in the following lemma.

\setlength{\parindent}{0em}
\begin{lemma}
{{\bf{(1)}}\,\,Let $T$ be a covering of $\Omega$, and let $R$ denote the set of all the maximal elements of $(T,\subseteq)$. Then, $R$ is a covering of $\Omega$, $(R,\subseteq)$ is an anti-chain, and it holds that $\omega_{T}=\omega_{R}$.

{\bf{(2)}}\,\,Let $T$ and $R$ be coverings of $\Omega$ such that both $(T,\subseteq)$ and $(R,\subseteq)$ are anti-chains. Further assume that for any $A\subseteq\Omega$, $\omega_{T}(A)=1\Longleftrightarrow\omega_{R}(A)=1$. Then, it holds that $T=R$.

{\bf{(3)}}\,\,Let $k\in[1,|\Omega|]$. Then, $\mathcal{P}(k,\Omega)$ is a covering of $\Omega$, $(\mathcal{P}(k,\Omega),\subseteq)$ is an anti-chain, and for any $A\subseteq\Omega$, it holds that $\omega_{\mathcal{P}(k,\Omega)}(A)=\left\lceil\frac{|A|}{k}\right\rceil$.
}
\end{lemma}

\begin{proof}
We note that (1) and (2) have been stated in [39, Propositions 1 and 2]. Since all the proofs are straightforward, the details are omitted.
\end{proof}

\setlength{\parindent}{2em}
Now we let $\left(H_{i}\mid i\in \Omega\right)$ be a family of finite abelian groups, and let $\mathbf{H}\triangleq\prod_{i\in \Omega}H_{i}$. Consider a covering $T$ of $\Omega$. For any codeword $\beta\in\mathbf{H}$, the $T$-weight of $\beta$ is defined as \begin{equation}\mbox{$\wt_{T}(\beta)=\omega_{T}(\supp(\beta))=\min\{|S|\mid S\subseteq T,\supp(\beta)\subseteq\bigcup_{I\in S}I\}$}.\end{equation}
It has been proven in \cite{15} that $d_{T}:\mathbf{H}\times\mathbf{H}\longrightarrow \mathbb{N}$ defined as $d_{T}(\alpha,\beta)=\wt_{T}(\alpha^{-1}\beta)$ induces a metric on $\mathbf{H}$, which is henceforth referred to as the $T$-combinatorial metric. Based on (1) of Lemma 2.3, one can assume that $(T,\subseteq)$ is an anti-chain without loss of generality. In addition, let $k\in[1,|\Omega|]$. Then, by (3) of Lemma 2.3 and (2.11), we have
\begin{equation}\forall~\beta\in\mathbf{H}:\wt_{\mathcal{P}(k,\Omega)}(\beta)=\left\lceil\frac{|\supp(\beta)|}{k}\right\rceil.\end{equation}
In particular, the $\mathcal{P}(1,\Omega)$-combinatorial metric is exactly the Hamming metric.

\setlength{\parindent}{2em}
Now we introduce partitions induced by combinatorial metric.

\setlength{\parindent}{0em}
\begin{notation}
For a covering $T$ of $\Omega$, we let $\mathcal{CO}(\mathbf{H},T)$ denote the partition of $\mathbf{H}$ such that for any $\beta,\theta\in\mathbf{H}$, $\beta\sim_{\mathcal{CO}(\mathbf{H},T)}\theta\Longleftrightarrow\wt_{T}(\beta)=\wt_{T}(\theta)$.
\end{notation}

\setlength{\parindent}{2em}
At the end of this subsection, we state the following conjecture proposed by Pinheiro, Machado and Firer in \cite{39} using our notation.

\setlength{\parindent}{0em}
\begin{conjecture}
Let $\mathbb{F}_2$ be the binary field and suppose that $\mathbf{H}={\mathbb{F}_2}^{\Omega}$. Consider $k\in[1,|\Omega|]$. Then, for any additive code $C\leqslant\mathbf{H}$ and $f\in\Hom(C,\mathbf{H})$ such that $\wt_{\mathcal{P}(k,\Omega)}(\alpha)=\wt_{\mathcal{P}(k,\Omega)}(f(\alpha))$ for all $\alpha\in C$, there exists $\varphi\in\Aut(\mathbf{H})$ such that $\varphi\mid_{C}=f$ and $\wt_{\mathcal{P}(k,\Omega)}(\alpha)=\wt_{\mathcal{P}(k,\Omega)}(\varphi(\alpha))$ for all $\alpha\in \mathbf{H}$. Alternatively speaking, for any $k\in[1,|\Omega|]$, $\mathcal{CO}(\mathbf{H},\mathcal{P}(k,\Omega))$ satisfies the MEP (see Definition 5.1).
\end{conjecture}

\setlength{\parindent}{2em}
We will prove in Section 5 that Conjecture 2.1 does not hold in general, as detailed in Theorem 5.3 and Remark 5.2. Our approach is based on the non-reflexivity of $\mathcal{CO}(\mathbf{H},\mathcal{P}(k,\Omega))$ and the relation between reflexivity and the MEP, as detailed in Section 4.2 and Section 5, respectively.

\section{Partitions induced by weighted poset metric}

\setlength{\parindent}{2em}
Throughout this and the next section, we let $\Omega$ be a nonempty finite set, and let $(G_{i}\mid i\in \Omega)$ and $(H_{i}\mid i\in \Omega)$ be two families of finite abelian groups such that $G_{i}\cong H_{i}$ and $|H_{i}|\triangleq h_{i}$ for all $i\in\Omega$. Write $$\mathbf{G}\triangleq\prod_{i\in \Omega}G_{i},~\mathbf{H}\triangleq\prod_{i\in \Omega}H_{i}.$$
For any $i\in\Omega$, let $\pi_{i}:G_{i}\times H_{i} \longrightarrow \mathbb{C}^{\ast}$ be a non-degenerate pairing. We define the non-degenerate pairing $f:\mathbf{G}\times\mathbf{H}\longrightarrow \mathbb{C}^{\ast}$ as
\begin{equation}\forall~(\alpha,\beta)\in\mathbf{G}\times\mathbf{H}: f(\alpha,\beta)=\prod_{i\in \Omega}\pi_{i}(\alpha_{i},\beta_{i}).\end{equation}
For any partition $\Gamma$ of $\mathbf{H}$, we let $\textbf{\textit{l}}(\Gamma)$ denote the left dual partition of $\Gamma$ with respect to $f$, as defined in (2.3). For any partition $\Delta$ of $\mathbf{G}$, we let $\textbf{\textit{r}}(\Delta)$ denote the right dual partition of $\Delta$ with respect to $f$, as defined in (2.4).

Throughout this section, we fix a poset $\mathbf{P}=(\Omega,\preccurlyeq_{\mathbf{P}})$. For any $D,I\subseteq\Omega$, define $\varphi(D,I)\in \mathds{Z}$ and $\psi(D,I)\in \mathds{Z}$ as follows:
\begin{eqnarray*}\hspace*{-12mm}\varphi(D,I)=\begin{cases}
(-1)^{|I\cap D|}\left(\prod_{i\in I-\max_{\mathbf{P}}(I)}h_{i}\right)\left(\prod_{i\in \max_{\mathbf{P}}(I)-D}(h_{i}-1)\right),&I\cap D\subseteq \max_{\mathbf{P}}(I);\\
0,&I\cap D\not\subseteq \max_{\mathbf{P}}(I),
\end{cases}
\end{eqnarray*}
\begin{eqnarray*}\hspace*{-12mm}\psi(D,I)=\begin{cases}
(-1)^{|I\cap D|}\left(\prod_{i\in D-\min_{\mathbf{P}}(D)}h_{i}\right)\left(\prod_{i\in \min_{\mathbf{P}}(D)-I}(h_{i}-1)\right),&I\cap D\subseteq \min_{\mathbf{P}}(D);\\
0,&I\cap D\not\subseteq \min_{\mathbf{P}}(D).
\end{cases}
\end{eqnarray*}
We also fix $\omega:\Omega\longrightarrow\mathbb{R}^{+}$, and define $\varpi:2^{\Omega}\longrightarrow\mathbb{R}$ as $\varpi(I)=\sum_{i\in I}\omega(i)$. Moreover, we write $\Lambda=\textbf{\textit{l}}(\mathcal{Q}(\mathbf{H},\mathbf{P},\omega))$, $\Theta=\textbf{\textit{r}}(\mathcal{Q}(\mathbf{G},\mathbf{\overline{P}},\omega))$.

\subsection{A sufficient condition for $\mathcal{Q}(\mathbf{H},\mathbf{P},\omega)$ to be reflexive}

\setlength{\parindent}{2em}
We begin by computing the left generalized Krawtchouk matrix of $(\Lambda,\mathcal{Q}(\mathbf{H},\mathbf{P},\omega))$ with respect to $f$. By [44, Proposition II.1], for any $\alpha\in \mathbf{G}$ and $I\in\mathcal{I}(\mathbf{P})$, we have
\begin{equation}\sum_{(\beta\in\mathbf{H},\langle \supp(\beta)\rangle_{\mathbf{P}}=I)}f(\alpha,\beta)=\varphi(\langle \supp(\alpha)\rangle_{\mathbf{\overline{P}}},I).\end{equation}
Hence by (2.8), for any $\alpha\in \mathbf{G}$ and $b\in\mathbb{R}$, it holds that
\begin{equation}\sum_{(\beta\in\mathbf{H},\wt_{(\mathbf{P},\omega)}(\beta)=b)}f(\alpha,\beta)=\sum_{(I\in\mathcal{I}(\mathbf{P}),\varpi(I)=b)}\varphi(\langle \supp(\alpha)\rangle_{\mathbf{\overline{P}}},I).\end{equation}
The right generalized Krawtchouk matrix of $(\mathcal{Q}(\mathbf{G},\mathbf{\overline{P}},\omega),\Theta)$ with respect to $f$ can be computed in a parallel fashion. More precisely, for any $\theta\in \mathbf{H}$ and $D\in\mathcal{I}(\mathbf{\overline{P}})$, we have
\begin{equation}\sum_{(\gamma\in\mathbf{G},\langle \supp(\gamma)\rangle_{\mathbf{\overline{P}}}=D)}f(\gamma,\theta)=\psi(D,\langle \supp(\theta)\rangle_{\mathbf{P}}).\end{equation}
Hence for any $\theta\in\mathbf{H}$ and $b\in\mathbb{R}$, it holds that
\begin{equation}\sum_{(\gamma\in\mathbf{G},\wt_{(\mathbf{\overline{P}},\omega)}(\gamma)=b)}f(\gamma,\theta)=\sum_{(D\in\mathcal{I}(\mathbf{\overline{P}}),\varpi(D)=b)}\psi(D,\langle \supp(\theta)\rangle_{\mathbf{P}}).\end{equation}
Using (3.3) and (3.5), we give the following sufficient conditions for two codewords to belong to the same member of $\Lambda$ or $\Theta$.

\setlength{\parindent}{0em}
\begin{proposition}
{Let $\lambda\in\Aut(\mathbf{P})$ such that $h_{i}=h_{\lambda(i)}$, $\omega(i)=\omega(\lambda(i))$ for all $i\in\Omega$. Then, we have:

{\bf{(1)}}\,\,For $\alpha,\gamma\in \mathbf{G}$ with $\langle \supp(\gamma)\rangle_{\mathbf{\overline{P}}}=\lambda[\langle \supp(\alpha)\rangle_{\mathbf{\overline{P}}}]$, it holds that $\alpha\sim_{\Lambda}\gamma$;

{\bf{(2)}}\,\,For $\beta,\theta\in \mathbf{H}$ with $\langle \supp(\theta)\rangle_{\mathbf{P}}=\lambda[\langle \supp(\beta)\rangle_{\mathbf{P}}]$, it holds that $\beta\sim_{\Theta}\theta$.
}
\end{proposition}

\begin{proof}
We only prove (1), and the proof of (2) is similar. Let $D=\langle \supp(\alpha)\rangle_{\mathbf{\overline{P}}}$. Then, we have $\langle \supp(\gamma)\rangle_{\mathbf{\overline{P}}}=\lambda[D]$. From $\lambda\in\Aut(\mathbf{P})$ and $h_{i}=h_{\lambda(i)}$ for all $i\in\Omega$, one can check that for any $I\subseteq\Omega$, it holds that $\varphi(D,I)=\varphi(\lambda[D],\lambda[I])$. Consider an arbitrary $b\in\mathbb{R}$. From $\lambda\in\Aut(\mathbf{P})$ and $\omega(i)=\omega(\lambda(i))$ for all $i\in\Omega$, one can check that for any $I\in\mathcal{I}(\mathbf{P})$ with $\varpi(I)=b$, it holds that $\lambda[I]\in\mathcal{I}(\mathbf{P})$, $\varpi(\lambda[I])=b$. Now by (3.3), we have
\begin{eqnarray*}
\begin{split}
\sum_{(\beta\in\mathbf{H},\wt_{(\mathbf{P},\omega)}(\beta)=b)}f(\gamma,\beta)&=\sum_{(I\in\mathcal{I}(\mathbf{P}),\varpi(I)=b)}\varphi(\lambda[D],I)\\
&=\sum_{(I\in\mathcal{I}(\mathbf{P}),\varpi(I)=b)}\varphi(\lambda[D],\lambda[I])\\
&=\sum_{(I\in\mathcal{I}(\mathbf{P}),\varpi(I)=b)}\varphi(D,I)\\
&=\sum_{(\beta\in\mathbf{H},\wt_{(\mathbf{P},\omega)}(\beta)=b)}f(\alpha,\beta),
\end{split}
\end{eqnarray*}
which immediately implies the desired result.
\end{proof}

\setlength{\parindent}{2em}
Now we prove the main result of this subsection.

\setlength{\parindent}{0em}
\begin{theorem}
{Assume that $(\mathbf{P},\omega)$ satisfies the UDP, and for any $u,v\in\Omega$ such that $\len_{\mathbf{P}}(u)=\len_{\mathbf{P}}(v)$ and $\omega(u)=\omega(v)$, it holds that $h_{u}=h_{v}$. Then, we have $\Lambda=\mathcal{Q}(\mathbf{G},\mathbf{\overline{P}},\omega)$, $\Theta=\mathcal{Q}(\mathbf{H},\mathbf{P},\omega)$, and both $\mathcal{Q}(\mathbf{G},\mathbf{\overline{P}},\omega)$ and $\mathcal{Q}(\mathbf{H},\mathbf{P},\omega)$ are reflexive.
}
\end{theorem}

\begin{proof}
First, consider $\alpha,\gamma\in\mathbf{G}$ with $\wt_{(\mathbf{\overline{P}},\omega)}(\alpha)=\wt_{(\mathbf{\overline{P}},\omega)}(\gamma)$. By (2.8), we have $\varpi(\langle\supp(\alpha)\rangle_{\mathbf{\overline{P}}})=\varpi(\langle\supp(\gamma)\rangle_{\mathbf{\overline{P}}})$. Since $(\mathbf{P},\omega)$ satisfies the UDP, by Lemma 2.2, we can choose $\lambda\in\Aut(\mathbf{P})$ such that $\langle\supp(\gamma)\rangle_{\mathbf{\overline{P}}}=\lambda[\langle\supp(\alpha)\rangle_{\mathbf{\overline{P}}}]$ and $\omega(i)=\omega(\lambda(i))$ for all $i\in\Omega$. For any $i\in\Omega$, it follows from the fact $\lambda\in\Aut(\mathbf{P})$ that $\len_{\mathbf{P}}(i)=\len_{\mathbf{P}}(\lambda(i))$, which, along with $\omega(i)=\omega(\lambda(i))$, implies that $h_{i}=h_{\lambda(i)}$. By Proposition 3.1, we have $\alpha\sim_{\Lambda}\gamma$. It follows that $\mathcal{Q}(\mathbf{G},\mathbf{\overline{P}},\omega)$ is finer than $\Lambda$. A similarly discussion leads to the fact that $\mathcal{Q}(\mathbf{H},\mathbf{P},\omega)$ is finer than $\Theta$. Therefore $(\mathcal{Q}(\mathbf{G},\mathbf{\overline{P}},\omega),\mathcal{Q}(\mathbf{H},\mathbf{P},\omega))$ is mutually dual with respect to $f$, which, along with Lemma 2.1, immediately implies the desired result.
\end{proof}

\subsection{The case that $\mathbf{P}$ is hierarchical}
\setlength{\parindent}{2em}
Throughout this subsection, we assume that $\omega$ is integer-valued, i.e.,
\begin{equation}\mbox{$\omega(i)\in\mathbb{Z}^{+}$ for all $i\in\Omega$}.\end{equation}
We also let $m$ be the largest cardinality of a chain in $\mathbf{P}$, and for any $j\in[1,m]$, let $W_{j}\triangleq\{u\in\Omega\mid\len_{\mathbf{P}}(u)=j\}$. Moreover, for any $D\subseteq\Omega$, we let $\sigma(D)$ denote the largest integer $r\in[1,m]$ such that $D\subseteq \bigcup_{j=r}^{m}W_{j}$.

As a generalization of [44, Notation II.1], we can relate each $\alpha\in\mathbf{G}$ with a polynomial $F(\omega,\alpha)$ defined as
\begin{equation}F(\omega,\alpha)\triangleq\sum_{l=0}^{\varpi(\Omega)}\sum_{(\beta\in \mathbf{H},\wt_{(\mathbf{P},\omega)}(\beta)=l)}f(\alpha,\beta)x^{l}.\end{equation}
By the definition of $\Lambda$, we infer that
\begin{equation}\forall~\alpha,\gamma\in\mathbf{G}:\alpha\sim_{\Lambda}\gamma\Longleftrightarrow F(\omega,\alpha)=F(\omega,\gamma).\end{equation}
In addition, we can derive a more explicit form of $F(\omega,\alpha)$, as detailed in the following proposition.

\setlength{\parindent}{0em}
\begin{proposition}
{\bf{(1)}}\,\,Let $\alpha\in \mathbf{G}$, and write $D=\langle \supp(\alpha)\rangle_{\mathbf{\overline{P}}}$, $X=(\Omega-D)\cup\min_{\mathbf{P}}(D)$. Then, we have
$$F(\omega,\alpha)=\sum_{(I\in\mathcal{I}(\mathbf{P}),I\subseteq X)}(-1)^{|I\cap D|}\left(\prod_{i\in I-\max_{\mathbf{P}}(I)}h_{i}\right)\left(\prod_{i\in \max_{\mathbf{P}}(I)-D}(h_{i}-1)\right)x^{\varpi(I)}.$$
In addition, if $h_{i}\geqslant2$ for all $i\in\Omega$, then $\deg(F(\omega,\alpha))=\varpi(X)$.

{\bf{(2)}}\,\,Suppose that $\mathbf{P}$ is hierarchical. Let $\alpha\in \mathbf{G}$, and write $D=\langle \supp(\alpha)\rangle_{\mathbf{\overline{P}}}$, $r=\sigma(D)$. Then, we have
\begin{eqnarray*}
\begin{split}
\hspace*{-8mm}F(\omega,\alpha)=&\left(\prod_{i\in \left(\bigcup_{j=1}^{r-1}W_{j}\right)}h_{i}\cdot x^{\omega(i)}\right)\left(\prod_{i\in W_{r}\cap D}\left(1-x^{\omega(i)}\right)\right)\left(\prod_{i\in W_{r}-D}\left((h_{i}-1)x^{\omega(i)}+1\right)\right)\\
&+\sum_{t=1}^{r-1}\left(\prod_{i\in \left(\bigcup_{j=1}^{t-1}W_{j}\right)}h_{i}\cdot x^{\omega(i)}\right)\left(\prod_{i\in W_{t}}\left((h_{i}-1)x^{\omega(i)}+1\right)\right)\\
&-\sum_{t=2}^{r}\left(\prod_{i\in \left(\bigcup_{j=1}^{t-1}W_{j}\right)}h_{i}\cdot x^{\omega(i)}\right).
\end{split}
\end{eqnarray*}
In addition, if $h_{i}\geqslant2$ for all $i\in\Omega$, then $\deg(F(\omega,\alpha))=\varpi(\bigcup_{j=1}^{r}W_{j})$.
\end{proposition}

\begin{proof}
{\bf{(1)}}\,\,For any $I\in\mathcal{I}(\mathbf{P})$, one can check that $I\cap D\subseteq \max_{\mathbf{P}}(I)\Longleftrightarrow I\subseteq X$. With such an observation, the first part is a direct consequence of (3.3) and (3.7). In addition, the second part follows from the first part and the fact that $X\in\mathcal{I}(\mathbf{P})$, as desired.

{\bf{(2)}}\,\,Define $g:\{(t,V)\mid t\in[1,m],V\subseteq W_{t},V\neq\emptyset\}\longrightarrow 2^{\Omega}$ as $g(t,V)=(\bigcup_{j=1}^{t-1}W_{j})\cup V$. From $\mathbf{P}$ is hierarchical, we infer that $g$ is injective and the range of $g$ is equal to $\mathcal{I}(\mathbf{P})-\{\emptyset\}$. Moreover, for any $t\in[1,m]$, $V\subseteq W_{t}$, $V\neq\emptyset$, we have $\max_{\mathbf{P}}(g(t,V))=V$. Again by $\mathbf{P}$ is hierarchical, together with $\sigma(D)=r$, we have $\min_{\mathbf{P}}(D)=W_{r}\cap D$, $(\Omega-D)\cup\min_{\mathbf{P}}(D)=\bigcup_{j=1}^{r}W_{j}$. By (1), we have
\begin{eqnarray*}
\begin{split}
&F(\omega,\alpha)-1\\
&=\sum_{t=1}^{r}\sum_{(V\subseteq W_{t},V\neq\emptyset)}(-1)^{|V\cap D|}\left(\prod_{i\in \left(\bigcup_{j=1}^{t-1}W_{j}\right)}h_{i}\right)\left(\prod_{i\in V-D}(h_{i}-1)\right)x^{\varpi\left(\left(\bigcup_{j=1}^{t-1}W_{j}\right)\cup V\right)}\\
&=\sum_{t=1}^{r}\left(\prod_{i\in \left(\bigcup_{j=1}^{t-1}W_{j}\right)}h_{i}\cdot x^{\omega(i)}\right)\left(\sum_{(V\subseteq W_{t},V\neq\emptyset)}(-1)^{|V\cap D|}\left(\prod_{i\in V-D}(h_{i}-1)\right)x^{\varpi(V)}\right)\\
&=\sum_{t=1}^{r}\left(\prod_{i\in \left(\bigcup_{j=1}^{t-1}W_{j}\right)}h_{i}\cdot x^{\omega(i)}\right)\left(\left(\prod_{i\in W_{t}\cap D}\left(1-x^{\omega(i)}\right)\right)\left(\prod_{i\in W_{t}-D}\left((h_{i}-1)x^{\omega(i)}+1\right)\right)-1\right).
\end{split}
\end{eqnarray*}
Since for any $t\in[1,r-1]$, it holds true that $W_{t}\cap D=\emptyset$, $W_{t}-D=W_{t}$, the first part immediately follows from the above computation. In addition, the second part immediately follows from (1), as desired.
\end{proof}

\setlength{\parindent}{2em}
We are in a position to derive a necessary and sufficient condition for two codewords of $\mathbf{G}$ to belong to the same member of $\Lambda$ when $\mathbf{P}$ is hierarchical.

\setlength{\parindent}{0em}
\begin{proposition}
{Assume that $h_{i}\geqslant2$ for all $i\in\Omega$ and $\mathbf{P}$ is hierarchical. Let $\alpha,\gamma\in \mathbf{G}$, and write $D=\langle \supp(\alpha)\rangle_{\mathbf{\overline{P}}}$, $B=\langle \supp(\gamma)\rangle_{\mathbf{\overline{P}}}$. Then, $\alpha\sim_{\Lambda}\gamma$ if and only if there exists $\lambda\in\Aut(\mathbf{P})$ such that $D=\lambda[B]$ and $h_{i}=h_{\lambda(i)}$, $\omega(i)=\omega(\lambda(i))$ for all $i\in \Omega$.
}
\end{proposition}

\begin{proof}
Since the ``if'' part follows from Proposition 3.1, it remains to establish the ``only if'' part. Suppose that $\alpha\sim_{\Lambda}\gamma$, and write $r=\sigma(D)$, $s=\sigma(B)$. By (3.8), we have $F(\omega,\alpha)=F(\omega,\gamma)$. From Proposition 3.2, we deduce that $\varpi(\bigcup_{j=1}^{r}W_{j})=\deg(F(\omega,\alpha))=\deg(F(\omega,\gamma))=\varpi(\bigcup_{j=1}^{s}W_{j})$, which implies that $r=s$. Now Proposition 3.2 further implies that
\begin{eqnarray*}
\begin{split}
&\left(\prod_{i\in W_{r}\cap D}(x^{\omega(i)}-1)\right)\left(\prod_{i\in W_{r}-D}(x^{\omega(i)}+(h_{i}-1)^{-1})\right)\\
&=\left(\prod_{i\in W_{r}\cap B}(x^{\omega(i)}-1)\right)\left(\prod_{i\in W_{r}-B}(x^{\omega(i)}+(h_{i}-1)^{-1})\right).
\end{split}
\end{eqnarray*}
By Proposition A.1, which we state and prove in the appendix, we can choose a bijection $\varepsilon:W_{r}\cap B\longrightarrow W_{r}\cap D$ such that $h_{i}=h_{\varepsilon(i)}$, $\omega(i)=\omega(\varepsilon(i))$ for all $i\in W_{r}\cap B$. Now we can further choose a permutation $\varepsilon_1$ of $W_{r}$ such that $\varepsilon_1\mid_{W_{r}\cap B}=\varepsilon$ and $h_{i}=h_{\varepsilon_1(i)}$, $\omega(i)=\omega(\varepsilon_1(i))$ for all $i\in W_{r}$. Define $\lambda:\Omega\longrightarrow\Omega$ as $\lambda\mid_{W_{r}}=\varepsilon_1$ and $\lambda\mid_{\Omega-W_{r}}=\id_{\Omega-W_{r}}$. Since $\mathbf{P}$ is hierarchical, it is straightforward to verify that $\lambda\in\Aut(\mathbf{P})$, $\lambda[B]=D$, and $h_{i}=h_{\lambda(i)}$, $\omega(i)=\omega(\lambda(i))$ for all $i\in \Omega$, as desired.
\end{proof}

\setlength{\parindent}{2em}
Now we give necessary and sufficient conditions for $\mathcal{Q}(\mathbf{H},\mathbf{P},\omega)$ to be reflexive when $\mathbf{P}$ is hierarchical. The following is the main result of this subsection.

\setlength{\parindent}{0em}
\begin{theorem}
{Assume that $h_{i}\geqslant2$ for all $i\in\Omega$ and $\mathbf{P}$ is hierarchical. Then, $\Lambda$ is finer than $\mathcal{Q}(\mathbf{G},\mathbf{\overline{P}},\omega)$. Moreover, the following four statements are equivalent to each other:

{\bf{(1)}}\,\,$(\mathbf{P},\omega)$ satisfies the UDP, and for any $u,v\in\Omega$ such that $\len_{\mathbf{P}}(u)=\len_{\mathbf{P}}(v)$ and $\omega(u)=\omega(v)$, it holds that $h_{u}=h_{v}$;

{\bf{(2)}}\,\,$(\mathcal{Q}(\mathbf{G},\mathbf{\overline{P}},\omega),\mathcal{Q}(\mathbf{H},\mathbf{P},\omega)$ is mutually dual with respect to $f$;

{\bf{(3)}}\,\,$\mathcal{Q}(\mathbf{H},\mathbf{P},\omega)$ is reflexive;

{\bf{(4)}}\,\,$\Lambda=\mathcal{Q}(\mathbf{G},\mathbf{\overline{P}},\omega)$.
}
\end{theorem}

\begin{proof}
First of all, it follows from Proposition 3.3 that $\Lambda$ is finer than $\mathcal{Q}(\mathbf{G},\mathbf{\overline{P}},\omega)$. Since $\mathcal{I}(\mathbf{\overline{P}})=\{\Omega-I\mid I\in\mathcal{I}(\mathbf{P})\}$ and $h_{i}\geqslant2$ for all $i\in\Omega$, we have $|\mathcal{Q}(\mathbf{G},\mathbf{\overline{P}},\omega)|=|\mathcal{Q}(\mathbf{H},\mathbf{P},\omega)|$. Now $(1)\Longrightarrow(2)$ follows from Theorem 3.1, and $(2)\Longrightarrow(3)$ follows from Lemma 2.1. Suppose that $\mathcal{Q}(\mathbf{H},\mathbf{P},\omega)$ is reflexive. Then, by Lemma 2.1, we have $|\Lambda|=|\mathcal{Q}(\mathbf{H},\mathbf{P},\omega)|=|\mathcal{Q}(\mathbf{G},\mathbf{\overline{P}},\omega)|$, which, along with the fact that $\Lambda$ is finer than $\mathcal{Q}(\mathbf{G},\mathbf{\overline{P}},\omega)$, implies that $\Lambda=\mathcal{Q}(\mathbf{G},\mathbf{\overline{P}},\omega)$, which further establishes $(3)\Longrightarrow(4)$. Therefore it remains to prove $(4)\Longrightarrow(1)$.

$(4)\Longrightarrow(1)$\,\,First, we let $D,B\in\mathcal{I}(\mathbf{\overline{P}})$ with $\varpi(D)=\varpi(B)$. Since $h_{i}\geqslant2$ for all $i\in\Omega$, we can choose $\alpha,\gamma\in\mathbf{G}$ such that $\langle \supp(\alpha)\rangle_{\mathbf{\overline{P}}}=D$, $\langle \supp(\gamma)\rangle_{\mathbf{\overline{P}}}=B$. From $\varpi(D)=\varpi(B)$, we infer that $\wt_{(\mathbf{\overline{P}},\omega)}(\alpha)=\wt_{(\mathbf{\overline{P}},\omega)}(\gamma)$, which further implies that $\alpha\sim_{\Lambda}\gamma$. By Proposition 3.3, we can choose $\lambda\in\Aut(\mathbf{P})$ such that $D=\lambda[B]$ and $h_{i}=h_{\lambda(i)}$, $\omega(i)=\omega(\lambda(i))$ for all $i\in \Omega$. It follows from Lemma 2.2 that $(\mathbf{P},\omega)$ satisfies the UDP. Next, we let $u,v\in\Omega$ such that $\len_{\mathbf{P}}(u)=\len_{\mathbf{P}}(v)$ and $\omega(u)=\omega(v)$. Consider $B_1=\langle\{u\}\rangle_{\mathbf{\overline{P}}}$, $D_1=\langle\{v\}\rangle_{\mathbf{\overline{P}}}$. Since $\mathbf{P}$ is hierarchical, it is straightforward to verify that $B_1,D_1\in\mathcal{I}(\mathbf{\overline{P}})$, $\varpi(B_1)=\varpi(D_1)$. Hence we can choose $\mu\in\Aut(\mathbf{P})$ such that $D_1=\mu[B_1]$ and $h_{i}=h_{\mu(i)}$ for all $i\in \Omega$. Since $\mu\in\Aut(\mathbf{P})$, we have $v=\mu(u)$, which further implies that $h_{u}=h_{v}$, as desired.
\end{proof}

\begin{remark}
If $\omega$ is the constant $1$ map, then Theorem 3.2 recovers [18, Theorem 5.5] and part of [18, Theorem 5.4].
\end{remark}

\section{Partitions induced by combinatorial metric}

\setlength{\parindent}{2em}
Throughout this section, for a covering $T$ of $\Omega$, we define $\omega_{T}:2^{\Omega}\longrightarrow\mathbb{N}$ as in (2.9), and for any $\alpha\in\mathbf{G}$, $\beta\in\mathbf{H}$, we let $\wt_{T}(\alpha)\triangleq\omega_{T}(\supp(\alpha))$, $\wt_{T}(\beta)\triangleq\omega_{T}(\supp(\beta))$, as in (2.11). For any polynomial $g\in\mathbb{C}[x]$, we let $g_{[i]}$ denote the coefficient of $x^{i}$ in $g$.

\subsection{Sufficient but not necessary conditions for reflexivity}

\setlength{\parindent}{2em}
In this subsection, we prove the following theorem.

\setlength{\parindent}{0em}
\begin{theorem}
{Assume that $h_{i}\geqslant2$ for all $i\in\Omega$. Let $T$ be a covering of $\Omega$ such that $(T,\subseteq)$ is an anti-chain. Then, the following three statements are equivalent to each other:

{\bf{(1)}}\,\,$\mathcal{CO}(\mathbf{G},T)$ is finer than $\textbf{\textit{l}}(\mathcal{CO}(\mathbf{H},T))$;

{\bf{(2)}}\,\,$\mathcal{CO}(\mathbf{G},T)=\textbf{\textit{l}}(\mathcal{CO}(\mathbf{H},T))$;

{\bf{(3)}}\,\,$T$ is a partition of $\Omega$, and $\prod_{i\in U}h_{i}=\prod_{j\in V}h_{j}$ for all $U,V\in T$.
}
\end{theorem}

\begin{proof}
First of all, since $h_{i}\geqslant2$ for all $i\in\Omega$, we have $|\mathcal{CO}(\mathbf{G},T)|=|\mathcal{CO}(\mathbf{H},T)|=|\{\omega_{T}(A)\mid A\subseteq \Omega\}|$, which, along with Lemma 2.1, implies that $(1)\Longleftrightarrow(2)$. Next, suppose that (2) holds true, and we will show that $T$ is a partition of $\Omega$. By way of contradiction, we assume that $T$ is not a partition of $\Omega$. Since $(T,\subseteq)$ is an anti-chain, we have $\emptyset\not\in T$. Hence we can choose $A,B\in T$ such that $A\cap B\neq\emptyset$, $B\nsubseteq A$. Therefore we can further choose $u\in B-A$, $v\in A\cap B$. Apparently, we have $\omega_{T}(\{u\})=\omega_{T}(\{u,v\})=1$. Since $h_{i}\geqslant2$ for all $i\in\Omega$, we can choose $\alpha,\gamma\in\mathbf{G}$ such that $\supp(\alpha)=\{u\}$, $\supp(\gamma)=\{u,v\}$. Applying (3.2) to the anti-chain $(\Omega,=)$, we have
$$a\triangleq\sum_{(\beta\in\mathbf{H},\wt_{T}(\beta)\leqslant1)}f(\alpha,\beta)=\sum_{(I\subseteq\Omega,\omega_{T}(I)\leqslant1)}(-1)^{|I\cap \{u\}|}\left(\prod_{i\in I-\{u\}}(h_{i}-1)\right),$$
$$b\triangleq\sum_{(\beta\in\mathbf{H},\wt_{T}(\beta)\leqslant1)}f(\gamma,\beta)=\sum_{(I\subseteq\Omega,\omega_{T}(I)\leqslant1)}(-1)^{|I\cap \{u,v\}|}\left(\prod_{i\in I-\{u,v\}}(h_{i}-1)\right).$$
By $\wt_{T}(\alpha)=\wt_{T}(\gamma)=1$ and (2), we have $\alpha\sim_{\textbf{\textit{l}}(\mathcal{CO}(\mathbf{H},T))}\gamma$, which further implies that $a=b$. On the other hand, some straightforward computation yields that
\begin{equation}a-b=\left(\sum_{(J\subseteq\Omega-\{u,v\},\omega_{T}(J\cup\{v\})\leqslant1,\omega_{T}(J\cup\{u,v\})\geqslant2)}~\prod_{i\in J}(h_{i}-1)\right)h_{v}.\end{equation}
Noticing that $u\not\in A$, we have $A-\{v\}\subseteq\Omega-\{u,v\}$. Since $v\in A$, $A\in T$, we have $(A-\{v\})\cup\{v\}=A$, $\omega_{T}(A)=1$. Again by $v\in A$, we have $(A-\{v\})\cup\{u,v\}=A\cup\{u\}$. If $\omega_{T}(A\cup\{u\})\leqslant1$, then we can choose $C\in T$ such that $A\cup\{u\}\subseteq C$, which, along with $u\not\in A$, further implies that $A\subsetneqq C$, which is impossible since $A,C\in T$, $(T,\subseteq)$ is an anti-chain. It then follows that $\omega_{T}(A\cup\{u\})\geqslant2$. By the above discussion and $h_{i}\geqslant2$ for all $i\in\Omega$, (4.1) immediately implies that $a-b\geqslant1$, a contradiction to $a=b$, as desired. Therefore we have shown that $T$ is a partition of $\Omega$.

\hspace*{4mm}\,\,Now we prove $(2)\Longleftrightarrow(3)$. By the discussion in the previous paragraph, we assume that $T$ is a partition of $\Omega$. Let $\mathbf{G}_1=\prod_{A\in T}(\prod_{i\in A}G_{i})$, $\mathbf{H}_1=\prod_{A\in T}(\prod_{i\in A}H_{i})$. For any $A\in T$, define the non-degenerate pairing $\varsigma_{A}:(\prod_{i\in A}G_{i})\times(\prod_{i\in A}H_{i})\longrightarrow\mathbb{C}^{*}$ as $\varsigma_{A}(\gamma,\theta)=\prod_{i\in A}\pi_{i}(\gamma_{i},\theta_{i})$. Moreover, define the non-degenerate pairing $f_1:\mathbf{G}_1\times\mathbf{H}_1\longrightarrow\mathbb{C}^{*}$ as $f_1(\lambda,\mu)=\prod_{A\in T}\varsigma_{A}(\lambda_{A},\mu_{A})$, and let $\widetilde{\omega}$ denote the constant $1$ map defined on $T$. Since $T$ is a partition of $\Omega$, some straightforward computation implies that $\mathcal{CO}(\mathbf{G},T)=\textbf{\textit{l}}(\mathcal{CO}(\mathbf{H},T))$ holds true if and only if $\mathcal{Q}(\mathbf{G}_1,(T,=),\widetilde{\omega})$ is the left dual partition of $\mathcal{Q}(\mathbf{H}_1,(T,=),\widetilde{\omega})$ with respect to $f_1$. For any $A\in T$, by $A\neq\emptyset$ and $h_{i}\geqslant2$ for all $i\in\Omega$, we have $|\prod_{i\in A}H_{i}|=\prod_{i\in A}h_{i}\geqslant2$. It follows from applying Theorem 3.2 to $\mathbf{G}_1$, $\mathbf{H}_1$, $f_1$ and $((T,=),\widetilde{\omega})$ that $\mathcal{Q}(\mathbf{G}_1,(T,=),\widetilde{\omega})$ is the left dual partition of $\mathcal{Q}(\mathbf{H}_1,(T,=),\widetilde{\omega})$ with respect to $f_1$ if and only if $\prod_{i\in U}h_{i}=\prod_{j\in V}h_{j}$ for all $U,V\in T$, which further concludes the proof of $(2)\Longleftrightarrow(3)$.
\end{proof}

\setlength{\parindent}{2em}
We remark that by Lemma 2.1, each of (1)--(3) of Theorem 4.1 is a sufficient condition for $\mathcal{CO}(\mathbf{H},T)$ to be reflexive. We will show in Section 4.2 that such sufficient conditions are not necessary.

\subsection{Non-reflexive partitions of the form $\mathcal{CO}(\mathbf{H},\mathcal{P}(k,\Omega))$}

\setlength{\parindent}{2em}
Throughout this subsection, we fix $q\in\mathbb{Z}^{+}$ such that $q\geqslant2$.

From now on, we will focus on the $\mathcal{P}(k,\Omega)$-combinatorial metric, where $k\in[1,|\Omega|]$. In addition, we will always assume that all the $H_{i}'s$ have order $q$. Such an additional assumption will enable us to relate the partitions with the well known Krawtchouk polynomials, which we first recall in the following definition (see, e.g., \cite{9,25,32}).

\setlength{\parindent}{0em}
\begin{definition}
{For any $(n,k)\in\mathbb{N}\times\mathbb{N}$, define the Krawtchouk polynomial $\mathbf{KU}_{(n,k)}$ as
$$\mathbf{KU}_{(n,k)}=\frac{(-1)^{k}}{k!}\sum_{t=0}^{k}\binom{k}{t}(q-1)^{k-t}\left(\prod_{i=0}^{t-1}(x-i)\right)\left(\prod_{i=0}^{k-t-1}(x-n+i)\right).$$
}
\end{definition}

\setlength{\parindent}{2em}
We collect all the properties of the Krawtchouk polynomials that we need in the following lemma.

\setlength{\parindent}{0em}
\begin{lemma}
{{\bf{(1)}}\,\,Let $(n,k)\in\mathbb{N}\times\mathbb{N}$. Then, we have $\deg(\mathbf{KU}_{(n,k)})=k$. Moreover, for any $s\in[0,n]$, it holds that
$$\mathbf{KU}_{(n,k)}(s)=\sum_{t=0}^{k}(-1)^{t}(q-1)^{k-t}\binom{s}{t}\binom{n-s}{k-t}=((1-x)^{s}(1+(q-1)x)^{n-s})_{[k]}.$$
{\bf{(2)}}\,\,Let $n\in\mathbb{Z}^{+}$, $k\in\mathbb{N}$. Then, for any $s\in[1,n]$, we have $\sum_{l=0}^{k}\mathbf{KU}_{(n,l)}(s)=\mathbf{KU}_{(n-1,k)}(s-1)$.

{\bf{(3)}}\,\,Suppose that $q=2$. Let $(n,k)\in\mathbb{N}\times\mathbb{N}$. Then, for any $s\in[0,n]$, it holds that $\mathbf{KU}_{(n,k)}(n-s)=(-1)^{k}\mathbf{KU}_{(n,k)}(s)$.

{\bf{(4)}}\,\,Let $n\in\mathbb{Z}^{+}$, $k\in[1,n]$. Then, $\mathbf{KU}_{(n,k)}$ has $k$ distinct roots in $\mathbb{R}$, all of which lie between $0$ and $n$. Moreover, ${\mathbf{KU}_{(n,k)}}^{'}$ has $k-1$ distinct roots in $\mathbb{R}$, all of which lie between the smallest root and the largest root of $\mathbf{KU}_{(n,k)}$.

{\bf{(5)}}\,\,Fix $k\in\mathbb{Z}^{+}$. For any $n\in\mathbb{N}$ such that $n\geqslant k$, let $u_{(n)}$ denote the smallest root of $\mathbf{KU}_{(n,k)}$. Then, the sequence $(\frac{u_{(n)}}{n}\mid n\in\mathbb{N},n\geqslant k)$ converges to $\frac{q-1}{q}$.
}
\end{lemma}

\begin{proof}
We note that (1)--(3) and the first part of (4) are well known and can be found in \cite{9,25}, and the second part of (4) follows from the first part of (4) along with the fact that $\deg(\mathbf{KU}_{(n,k)})=k$. Hence it remains to establish (5). Fix $k\in\mathbb{Z}^{+}$, and let $T=\{\lambda=(\lambda_0,\dots,\lambda_{k-1})\in\mathbb{R}^{k}\mid\sum_{i=0}^{k-1}{\lambda_{i}}^{2}=1\}$. For any $n\in\mathbb{N}$ such that $n\geqslant k$, let $c_{(n)}$ denote the following real number
$$\max\left\{(q-2)\left(\sum_{i=0}^{k-1}i{\lambda_{i}}^{2}\right)+2\sqrt{q-1}\left(\sum_{i=0}^{k-2}\lambda_{i}\lambda_{i+1}\sqrt{(i+1)(n-i)}\right)\mid\lambda\in T\right\}.$$
By [25, Theorem 6.1], for any $n\in\mathbb{N}$ such that $n\geqslant k$, we have $\frac{u_{(n)}}{n}=\frac{q-1}{q}-\frac{c_{(n)}}{qn}$. For an arbitrary $n\in\mathbb{N}$ such that $n\geqslant k$, some straightforward computation yields that
$$-2(k-1)\sqrt{(q-1)(k-1)n}\leqslant c_{(n)}\leqslant(q-2)(k-1)+2(k-1)\sqrt{(q-1)(k-1)n}.$$
Consequently, the sequence $(\frac{c_{(n)}}{qn}\mid n\in\mathbb{N},n\geqslant k)$ converges to $0$, which immediately implies the desired result.
\end{proof}

\setlength{\parindent}{2em}
Next, we characterize $\textbf{\textit{l}}(\mathcal{CO}(\mathbf{H},\mathcal{P}(k,\Omega)))$ in terms of the Krawtchouk polynomials, and give some sufficient conditions for $\mathcal{CO}(\mathbf{H},\mathcal{P}(k,\Omega))$ to be non-reflexive in terms of the Krawtchouk polynomials, especially in terms of their roots.

\setlength{\parindent}{0em}
\begin{proposition}
{Suppose that $h_{i}=q$ for all $i\in\Omega$. Fix $k\in[1,|\Omega|]$, and let $\Lambda=\textbf{\textit{l}}(\mathcal{CO}(\mathbf{H},\mathcal{P}(k,\Omega)))$. Then, the following five statements hold true:

{\bf{(1)}}\,\,$|\mathcal{CO}(\mathbf{H},\mathcal{P}(k,\Omega))|=\lceil\frac{|\Omega|}{k}\rceil+1$, $|\Lambda|\geqslant\lceil\frac{|\Omega|}{k}\rceil+1$, $\{1_{\mathbf{G}}\}\in\Lambda$. Moreover, $\mathcal{CO}(\mathbf{H},\mathcal{P}(k,\Omega))$ is non-reflexive if and only if $|\Lambda|\geqslant\frac{|\Omega|}{k}+2$;

{\bf{(2)}}\,\,Let $\alpha,\gamma\in\mathbf{G}$, $\alpha\neq1_{\mathbf{G}}$, $\gamma\neq1_{\mathbf{G}}$, and write $|\supp(\alpha)|=t$, $|\supp(\gamma)|=r$. Then, $\alpha\sim_{\Lambda}\gamma$ if and only if for any $s\in[1,|\Omega|-1]$ such that $k\mid s$, it holds that $\mathbf{KU}_{(|\Omega|-1,s)}(t-1)=\mathbf{KU}_{(|\Omega|-1,s)}(r-1)$;

{\bf{(3)}}\,\,Let $s\in[1,|\Omega|-1]$ such that $k\mid s$. Then, it holds that
$$|\Lambda|\geqslant|\{\mathbf{KU}_{(|\Omega|-1,s)}(j)\mid j\in[0,|\Omega|-1]\}|+1.$$
Further assume that $|\{\mathbf{KU}_{(|\Omega|-1,s)}(j)\mid j\in[0,|\Omega|-1]\}|-1\geqslant\frac{|\Omega|}{k}$. Then, $\mathcal{CO}(\mathbf{H},\mathcal{P}(k,\Omega))$ is non-reflexive;

{\bf{(4)}}\,\,Let $s\in[1,|\Omega|-1]$ such that $k\mid s$, and let $u$ denote the smallest root of $\mathbf{KU}_{(|\Omega|-1,s)}$. Assume that $\lfloor u\rfloor\geqslant\frac{|\Omega|}{k}$. Then, $\mathcal{CO}(\mathbf{H},\mathcal{P}(k,\Omega))$ is non-reflexive;

{\bf{(5)}}\,\,Suppose that $|\Omega|\geqslant3$. Let $s\in[2,|\Omega|-1]$ such that $k\mid s$, and let $w$ denote the smallest root of ${\mathbf{KU}_{(|\Omega|-1,s)}}^{'}$. Further assume that $\lfloor w\rfloor\geqslant\frac{|\Omega|}{k}$. Then, $\mathcal{CO}(\mathbf{H},\mathcal{P}(k,\Omega))$ is non-reflexive.
}
\end{proposition}

\begin{proof}
{\bf{(1)}}\,\,First, by (2.12) and the fact that $h_{i}\geqslant2$ for all $i\in\Omega$, we have
\begin{eqnarray*}
\begin{split}
|\mathcal{CO}(\mathbf{H},\mathcal{P}(k,\Omega))|&=\left|\left\{\left\lceil\frac{|\supp(\beta)|}{k}\right\rceil\mid \beta\in\mathbf{H}\right\}\right|=\left|\left\{\left\lceil\frac{s}{k}\right\rceil\mid s\in[0,|\Omega|]\right\}\right|\\
&=\left|[0,\left\lceil\frac{|\Omega|}{k}\right\rceil]\right|=\left\lceil\frac{|\Omega|}{k}\right\rceil+1.
\end{split}
\end{eqnarray*}
Now the rest immediately follows from Lemma 2.1.

{\bf{(2)}}\,\,First, we consider $\alpha\in\mathbf{G}-\{1_{\mathbf{G}}\}$ with $t=|\supp(\alpha)|$. Applying Proposition 3.2 to the anti-chain $(\Omega,=)$ and the constant $1$ map, we have
$$\sum_{l=0}^{|\Omega|}\sum_{(\beta\in \mathbf{H},|\supp(\beta)|=l)}f(\alpha,\beta)x^{l}=(1-x)^{t}(1+(q-1)x)^{|\Omega|-t},$$
which, along with (1) of Lemma 4.1, implies that for any $l\in\mathbb{N}$,
\begin{equation}\sum_{(\beta\in \mathbf{H},|\supp(\beta)|=l)}f(\alpha,\beta)=((1-x)^{t}(1+(q-1)x)^{|\Omega|-t})_{[l]}=\mathbf{KU}_{(|\Omega|,l)}(t).\end{equation}
For an arbitrary $b\in\mathbb{N}$, (2.12) implies that $(\forall~\beta\in\mathbf{H}:\wt_{\mathcal{P}(k,\Omega)}(\beta)\leqslant b\Longleftrightarrow|\supp(\beta)|\leqslant bk)$, which, in combination with (4.2) and (2) of Lemma 4.1, further implies that
\begin{eqnarray*}
\begin{split}
\sum_{(\beta\in \mathbf{H},\wt_{\mathcal{P}(k,\Omega)}(\beta)\leqslant b)}f(\alpha,\beta)&=\sum_{(\beta\in \mathbf{H},|\supp(\beta)|\leqslant bk)}f(\alpha,\beta)\\
&=\sum_{l=0}^{bk}\sum_{(\beta\in \mathbf{H},|\supp(\beta)|=l)}f(\alpha,\beta)\\
&=\sum_{l=0}^{bk}\mathbf{KU}_{(|\Omega|,l)}(t)=\mathbf{KU}_{(|\Omega|-1,bk)}(t-1).
\end{split}
\end{eqnarray*}
Now for $\alpha,\gamma\in\mathbf{G}-\{1_{\mathbf{G}}\}$ with $t=|\supp(\alpha)|$, $r=|\supp(\gamma)|$, from the above discussion and the definition of $\Lambda$, we deduce that
\begin{eqnarray*}
\begin{split}
\alpha\sim_{\Lambda}\gamma&\Longleftrightarrow\left(\forall~b\in\mathbb{N}:\sum_{(\beta\in\mathbf{H},\wt_{\mathcal{P}(k,\Omega)}(\beta)\leqslant b)}f(\alpha,\beta)=\sum_{(\beta\in\mathbf{H},\wt_{\mathcal{P}(k,\Omega)}(\beta)\leqslant b)}f(\gamma,\beta)\right)\\
&\Longleftrightarrow(\forall~b\in\mathbb{N}:\mathbf{KU}_{(|\Omega|-1,bk)}(t-1)=\mathbf{KU}_{(|\Omega|-1,bk)}(r-1)).
\end{split}
\end{eqnarray*}
By (1) of Lemma 4.1, we have $\mathbf{KU}_{(|\Omega|-1,0)}(t-1)=\mathbf{KU}_{(|\Omega|-1,0)}(r-1)=1$ and $\mathbf{KU}_{(|\Omega|-1,s)}(t-1)=\mathbf{KU}_{(|\Omega|-1,s)}(r-1)=0$ for all $s\in\mathbb{N}$ with $s\geqslant|\Omega|$, which immediately implies the desired result.

{\bf{(3)}}\,\,The first part follows from (2) and the fact that $\{1_{\mathbf{G}}\}\in\Lambda$, and the second part follows from (1) and the first part, as desired.

{\bf{(4)}}\,\,By (4) of Lemma 4.1, we have $|\{\mathbf{KU}_{(|\Omega|-1,s)}(j)\mid j\in[0,|\Omega|-1]\}|\geqslant\lfloor u\rfloor+1$, which, together with (3), immediately implies the desired result.

{\bf{(5)}}\,\,By (4) of Lemma 4.1, we have $|\{\mathbf{KU}_{(|\Omega|-1,s)}(j)\mid j\in[0,|\Omega|-1]\}|\geqslant\lfloor w\rfloor+1$, and hence the desired result again follows from (3).
\end{proof}

\setlength{\parindent}{2em}
As a first application of Proposition 4.1, we show that the sufficient conditions for reflexivity given in Theorem 4.1 are not necessary.

\setlength{\parindent}{0em}
\begin{proposition}
Suppose that $h_{i}=q=2$ for all $i\in\Omega$.

{\bf{(1)}}\,\,Assume that $|\Omega|\geqslant2$, and let $\Lambda=\textbf{\textit{l}}(\mathcal{CO}(\mathbf{H},\mathcal{P}(|\Omega|-1,\Omega)))$. Then, for any $\alpha,\gamma\in\mathbf{G}-\{1_{\mathbf{G}}\}$, $\alpha\sim_{\Lambda}\gamma$ if and only if $|\supp(\alpha)|\equiv|\supp(\gamma)|~(\bmod~2)$. Moreover, $\mathcal{CO}(\mathbf{H},\mathcal{P}(|\Omega|-1,\Omega))$ is reflexive;

{\bf{(2)}}\,\,Assume that $|\Omega|\geqslant2$, and let $\Lambda=\textbf{\textit{l}}(\mathcal{CO}(\mathbf{H},\mathcal{P}(2,\Omega)))$. Then, for any $\alpha,\gamma\in\mathbf{G}$, $\alpha\sim_{\Lambda}\gamma$ if and only if either $|\supp(\alpha)|=|\supp(\gamma)|$ or $|\supp(\alpha)|+|\supp(\gamma)|=|\Omega|+1$ holds true. Moreover, $\mathcal{CO}(\mathbf{H},\mathcal{P}(2,\Omega))$ is reflexive.
\end{proposition}

\begin{proof}
{\bf{(1)}}\,\,Let $\alpha,\gamma\in\mathbf{G}-\{1_{\mathbf{G}}\}$, and let $t=|\supp(\alpha)|$, $r=|\supp(\gamma)|$. By (2) of Proposition 4.1 and (1) of Lemma 4.1, we have $\alpha\sim_{\Lambda}\gamma\Longleftrightarrow\mathbf{KU}_{(|\Omega|-1,|\Omega|-1)}(t-1)=\mathbf{KU}_{(|\Omega|-1,|\Omega|-1)}(r-1)\Longleftrightarrow(-1)^{t-1}=(-1)^{r-1}\Longleftrightarrow t\equiv r~(\bmod~2)$, as desired. It then follows from (1) of Proposition 4.1 that $|\mathcal{CO}(\mathbf{H},\mathcal{P}(|\Omega|-1,\Omega))|=|\Lambda|=3$, which implies that $\mathcal{CO}(\mathbf{H},\mathcal{P}(|\Omega|-1,\Omega))$ is reflexive, as desired.

\hspace*{2mm}\,\,{\bf{(2)}}\,\,Let $\Delta$ denote the partition of $\mathbf{G}$ such that for any $\alpha,\gamma\in\mathbf{G}$, $\alpha\sim_{\Delta}\gamma$ if and only if either $|\supp(\alpha)|=|\supp(\gamma)|$ or $|\supp(\alpha)|+|\supp(\gamma)|=|\Omega|+1$ holds true. Let $\alpha,\gamma\in\mathbf{G}$ such that $\alpha\sim_{\Delta}\gamma$, and let $t=|\supp(\alpha)|$, $r=|\supp(\gamma)|$. We will show that $\alpha\sim_{\Lambda}\gamma$. If $t=r$, then $\alpha\sim_{\Lambda}\gamma$ follows from (2) of Proposition 4.1. Hence in the following, we assume that $t+r=|\Omega|+1$. By (3) of Lemma 4.1, for any $s\in[1,|\Omega|-1]$ such that $2\mid s$, we have $\mathbf{KU}_{(|\Omega|-1,s)}(t-1)=\mathbf{KU}_{(|\Omega|-1,s)}(r-1)$, which, along with (2) of Proposition 4.1, implies that $\alpha\sim_{\Lambda}\gamma$, as desired. It then follows that $\Delta$ is finer than $\Lambda$. Also noticing that $|\Delta|=|\mathcal{CO}(\mathbf{H},\mathcal{P}(2,\Omega))|=\lceil\frac{|\Omega|}{2}\rceil+1$, from Lemma 2.1, we infer that $\Lambda=\Delta$ and  $\mathcal{CO}(\mathbf{H},\mathcal{P}(2,\Omega))$ is reflexive, as desired.
\end{proof}

\setlength{\parindent}{0em}
\begin{remark}
{If $|\Omega|\geqslant3$, then neither $\mathcal{P}(|\Omega|-1,\Omega)$ nor $\mathcal{P}(2,\Omega)$ is a partition of $\Omega$. Hence Proposition 4.2 gives sufficient conditions for reflexivity which are not covered by those presented in Theorem 4.1.
}
\end{remark}

\setlength{\parindent}{2em}
Now we give some criterions for non-reflexivity.

\setlength{\parindent}{0em}
\begin{proposition}
{Suppose that $h_{i}=q$ for all $i\in\Omega$.

{\bf{(1)}}\,\,Assume that $q\geqslant3$, $|\Omega|\geqslant3$. Fix $k\in[2,|\Omega|-1]$ such that $|\Omega|\equiv1~(\bmod~k)$, and let $\Lambda=\textbf{\textit{l}}(\mathcal{CO}(\mathbf{H},\mathcal{P}(k,\Omega)))$. Then, for any $\alpha,\gamma\in\mathbf{G}$, $\alpha\sim_{\Lambda}\gamma$ if and only if $|\supp(\alpha)|=|\supp(\gamma)|$. Consequently, $\mathcal{CO}(\mathbf{H},\mathcal{P}(k,\Omega))$ is non-reflexive;

{\bf{(2)}}\,\,If $q\geqslant3$, $|\Omega|\geqslant4$, then $\mathcal{CO}(\mathbf{H},\mathcal{P}(|\Omega|-2,\Omega))$ is non-reflexive;

{\bf{(3)}}\,\,Assume that $q=2$, $|\Omega|\geqslant5$. Then, for any $k\in[\lceil\frac{|\Omega|}{2}\rceil,|\Omega|-2]$, $\mathcal{CO}(\mathbf{H},\mathcal{P}(k,\Omega))$ is non-reflexive;

{\bf{(4)}}\,\,Assume that $q=2$, $|\Omega|\geqslant7$. Then, for any $k\in[\lceil\frac{|\Omega|}{5}\rceil,|\Omega|-2]$ such that $2\nmid k$, $\mathcal{CO}(\mathbf{H},\mathcal{P}(k,\Omega))$ is non-reflexive.
}
\end{proposition}

\begin{proof}
{\bf{(1)}}\,\,Let $\alpha,\gamma\in\mathbf{G}$, and let $t=|\supp(\alpha)|$, $r=|\supp(\gamma)|$. By (2) of Proposition 4.1, we only need to prove the ``only if'' part. Suppose that $\alpha\sim_{\Lambda}\gamma$. If $1_{\mathbf{G}}\in\{\alpha,\gamma\}$, then by $\{1_{\mathbf{G}}\}\in\Lambda$, we have $\alpha=\gamma=1_{\mathbf{G}}$, and hence $t=r=0$, as desired. Therefore in the following, we assume that $\alpha\neq1_{\mathbf{G}}$, $\gamma\neq1_{\mathbf{G}}$. By $|\Omega|\equiv1~(\bmod~k)$ and (2) of Proposition 4.1, we have $\mathbf{KU}_{(|\Omega|-1,|\Omega|-1)}(t-1)=\mathbf{KU}_{(|\Omega|-1,|\Omega|-1)}(r-1)$, which, along with (1) of Lemma 4.1, implies that $(-1)^{t-1}(q-1)^{|\Omega|-t}=(-1)^{r-1}(q-1)^{|\Omega|-r}$. Since $q\geqslant3$, we have $|\Omega|-t=|\Omega|-r$, and hence $t=r$, as desired. It follows that $\Lambda$ is the partition induced by Hamming weight, which, together with $|\Omega|\geqslant3$, $k\geqslant2$, implies that $|\Lambda|=|\Omega|+1>\frac{|\Omega|}{k}+2$. Now the non-reflexivity of $\mathcal{CO}(\mathbf{H},\mathcal{P}(k,\Omega))$ immediately follows from (1) of Proposition 4.1, as desired.

\hspace*{2mm}\,\,{\bf{(2)}}\,\,Write $n=|\Omega|$. Since $q\geqslant3$, $n\geqslant4$, one can check that $\mathbf{KU}_{(n-1,n-2)}$ takes different values on $0$, $n-2$, $n-1$, respectively, which further implies that $|\{\mathbf{KU}_{(n-1,n-2)}(j)\mid j\in[0,n-1]\}|-1\geqslant2\geqslant\frac{n}{n-2}$. It then follows from (3) of Proposition 4.1 that $\mathcal{CO}(\mathbf{H},\mathcal{P}(n-2,\Omega))$ is non-reflexive, as desired.

\hspace*{2mm}\,\,{\bf{(3)}} and {\bf{(4)}}\,\,Write $n=|\Omega|$. Suppose that $n\geqslant5$, and fix $k\in[3,n-2]$. It follows from (1) of Lemma 4.1 and some straightforward computation that
$$\mathbf{KU}_{(n-1,k)}(0)=\binom{n-1}{k}\triangleq a,~\mathbf{KU}_{(n-1,k)}(1)=\binom{n-2}{k}-\binom{n-2}{k-1}\triangleq b,$$
$$\mathbf{KU}_{(n-1,k)}(2)=\binom{n-3}{k}-2\binom{n-3}{k-1}+\binom{n-3}{k-2}\triangleq c,$$
$$\mathbf{KU}_{(n-1,k)}(3)=\binom{n-4}{k}-3\binom{n-4}{k-1}+3\binom{n-4}{k-2}-\binom{n-4}{k-3}\triangleq d.$$
It is straightforward to verify the following facts:
\begin{equation}\hspace*{-2mm}a>|b|,~a>|c|,~a>|d|,~(b=c\Longleftrightarrow n=2k),~(b=d\Longleftrightarrow n=2k+1).\end{equation}
From (4.3), we infer that $|\{a,b,c,d\}|\geqslant3$. If $k\in[\lceil\frac{n}{2}\rceil,n-2]$, then we have $|\{\mathbf{KU}_{(n-1,k)}(j)\mid j\in[0,n-1]\}|-1\geqslant2\geqslant\frac{n}{k}$, which, along with (3) of Proposition 4.1, implies that $\mathcal{CO}(\mathbf{H},\mathcal{P}(k,\Omega))$ is non-reflexive, which further establishes (3). Hence it remains to prove (4). From now on, we assume that $n\geqslant7$, $k\in[\lceil\frac{n}{5}\rceil,n-2]$, $2\nmid k$. It follows from (3) of Lemma 4.1 that $\mathbf{KU}_{(n-1,k)}(n-1)=-a$, $\mathbf{KU}_{(n-1,k)}(n-2)=-b$, $\mathbf{KU}_{(n-1,k)}(n-3)=-c$, $\mathbf{KU}_{(n-1,k)}(n-4)=-d$. From (4.3), we infer that $|\{\pm a,b,c,d\}|\geqslant4$. Hence if $k\geqslant\frac{n}{3}$, then we have $|\{\mathbf{KU}_{(n-1,k)}(j)\mid j\in[0,n-1]\}|-1\geqslant3\geqslant\frac{n}{k}$, and the desired result follows from (3) of Proposition 4.1. Therefore in the following, we assume that $k\leqslant\frac{n-1}{3}$. By straightforward computation, we have $b+c=2\binom{n-3}{k}-2\binom{n-3}{k-1}$, $b-c=2\binom{n-3}{k-1}-2\binom{n-3}{k-2}$. Since $n\geqslant7$, we have $\frac{n-1}{3}\leqslant\frac{n-3}{2}$, and hence $k\leqslant\frac{n-3}{2}$. It immediately follows that $b+c>0$, $b-c>0$, and hence $b>\pm c$. The above discussion yields that
\begin{equation}a>b>\pm c>-b>-a.\end{equation}
From (4.4), we infer that $|\{\pm a,\pm b,\pm c\}|\geqslant5$. Hence if $k\geqslant\frac{n}{4}$, then we have $|\{\mathbf{KU}_{(n-1,k)}(j)\mid j\in[0,n-1]\}|-1\geqslant4\geqslant\frac{n}{k}$, and the desired result follows from (3) of Proposition 4.1. Therefore in the following, we further assume that $k\leqslant\frac{n-1}{4}$. Then, some straightforward computation yields that $c\neq0$, which, along with (4.4), implies that $|\{\pm a,\pm b,\pm c\}|=6$. It then follows from $k\geqslant\frac{n}{5}$ that $|\{\mathbf{KU}_{(n-1,k)}(j)\mid j\in[0,n-1]\}|-1\geqslant5\geqslant\frac{n}{k}$, and hence (3) of Proposition 4.1 concludes the proof.
\end{proof}

\setlength{\parindent}{2em}
The following theorem is the main result of this subsection. It includes the largest amount of non-reflexive partitions obtained in this section.

\setlength{\parindent}{0em}
\begin{theorem}
{Let $X$ be a finite abelian group with $|X|=q$. Then, the following four statements hold true:

{\bf{(1)}}\,\,Fix $k\in\mathbb{Z}^{+}$ such that $k\geqslant2$, $(k,q)\neq(2,2)$. Then, there exists $m\in\mathbb{Z}^{+}$ such that for any $n\in\mathbb{Z}^{+}$ with $n\geqslant m$, $n\geqslant k+1$, the partition $\mathcal{CO}(X^{n},\mathcal{P}(k,[1,n]))$ is non-reflexive;

{\bf{(2)}}\,\,If $q\geqslant3$, then for any $n\in\mathbb{Z}^{+}$ such that $n\geqslant3$, $\mathcal{CO}(X^{n},\mathcal{P}(2,[1,n]))$ is non-reflexive;

{\bf{(3)}}\,\,Let $n\in\mathbb{Z}^{+}$ such that one of the following three conditions holds:

\hspace*{6mm}\,\,{\bf{3.1)}}\,\,$3\mid n$ and
$$n\geqslant\frac{9(q-1)+\sqrt{48q^{4}-144q^{3}+189q^{2}-162q+81}}{2(2q-3)^{2}}+3;$$
\hspace*{6mm}\,\,{\bf{3.2)}}\,\,$n\equiv1~(\bmod~3)$, $n\geqslant4$, $q\geqslant3$;

\hspace*{6mm}\,\,{\bf{3.3)}}\,\,$n\equiv2~(\bmod~3)$ and
$$n\geqslant\frac{4q^{2}+3q-9+\sqrt{48q^{4}-72q^{3}+9q^{2}-54q+81}}{2(2q-3)^{2}}+3.$$
Then, it holds that $n\geqslant4$ and $\mathcal{CO}(X^{n},\mathcal{P}(3,[1,n]))$ is non-reflexive;

{\bf{(4)}}\,\,If $q=2$, then for any $n\in\mathbb{Z}^{+}$ such that $n\geqslant5$, $\mathcal{CO}(X^{n},\mathcal{P}(3,[1,n]))$ is non-reflexive.
}
\end{theorem}

\begin{proof}
{\bf{(1)}}\,\,For any $n\in\mathbb{N}$ such that $n\geqslant k$, let $u_{(n)}$ denote the smallest root of $\mathbf{KU}_{(n,k)}$. By (5) of Lemma 4.1, the sequence $(\lfloor u_{(n-1)}\rfloor/n\mid n\in\mathbb{N},n\geqslant k+1)$ converges to $\frac{q-1}{q}$. Since $k,q\geqslant2$, $(k,q)\neq(2,2)$, we have $\frac{q-1}{q}>\frac{1}{k}$. Hence we can choose $m\in\mathbb{Z}^{+}$ such that for any $n\in\mathbb{Z}^{+}$ with $n\geqslant m$, $n\geqslant k+1$, it holds that $\lfloor u_{(n-1)}\rfloor>n/k$. Now for any $n\in\mathbb{Z}^{+}$ such that $n\geqslant m$, $n\geqslant k+1$, applying (4) of Proposition 4.1 to $[1,n]$ and $X^{n}$ leads to the non-reflexivity of $\mathcal{CO}(X^{n},\mathcal{P}(k,[1,n]))$, which further establishes (1).

\hspace*{2mm}\,\,{\bf{(2)}}\,\,Suppose that $q\geqslant3$. We consider $n\in\mathbb{Z}^{+}$, $n\geqslant3$. It follows from some straightforward computation that the only root of ${\mathbf{KU}_{(n-1,2)}}^{'}$ is equal to $v=\frac{q-1}{q}n-\frac{3}{2}+\frac{2}{q}$. If $2\nmid n$, then the desired result follows from (1) of Proposition 4.3; if $2\mid n$, $n\geqslant5$, then along with $q\geqslant3$, one can readily verify that $\lfloor v\rfloor\geqslant\frac{n}{2}$, and hence the desired result follows from (5) of Proposition 4.1; and if $n=4$, then by $q\geqslant3$ and (2) of Proposition 4.3, the partition $\mathcal{CO}(X^{4},\mathcal{P}(2,[1,4]))$ is non-reflexive, which concludes the proof of (2).

\hspace*{2mm}\,\,{\bf{(3)}}\,\,Apparently, we have $n\geqslant4$. Let $w$ denote the smallest root of ${\mathbf{KU}_{(n-1,3)}}^{'}$. Some straightforward computation yields that
\begin{equation}w=\frac{q-1}{q}n-2+\frac{3}{q}-\frac{\sqrt{(q-1)(n-3)+\frac{1}{3}q^{2}}}{q}.\end{equation}
If either 3.1) or 3.3) holds true, then from (4.5) and some straightforward computation, we have $\lfloor w\rfloor\geqslant\frac{n}{3}$, and hence the desired result follows from (5) of Proposition 4.1; and if 3.2) holds true, then the desired result follows from (1) of Proposition 4.3, which further establishes (3).

\hspace*{2mm}\,\,{\bf{(4)}}\,\,Suppose that $q=2$. Let $n\in\mathbb{Z}^{+}$, $n\geqslant5$, and let $w$ denote the smallest root of ${\mathbf{KU}_{(n-1,3)}}^{'}$. If $n\in\{5,6\}$, then the desired result follows from (3) of Proposition 4.3; if $n\in[7,15]$, then the desired result follows from (4) of Proposition 4.3; if $n\geqslant16$, $n\not\equiv1~(\bmod~3)$, then we note that $n$ satisfies either 3.1) or 3.3), and hence the desired result follows from (3); if $n\geqslant20$, $n\equiv1~(\bmod~3)$, then by (4.5), we have $\lfloor w\rfloor\geqslant\frac{n}{3}$, and hence the desired result follows from (5) of Proposition 4.1; and if $n\in\{16,19\}$, then some straightforward computation yields that $|\{\mathbf{KU}_{(15,3)}(j)\mid j\in[0,6]\}|=7$, $|\{\mathbf{KU}_{(18,3)}(j)\mid j\in[0,7]\}|=8$, and hence the desired result follows from (3) of Proposition 4.1, which further concludes the proof of (4).
\end{proof}

\begin{remark}
In Section 5, we will use Proposition 4.3 and Theorem 4.2 to provide counter-examples to Conjecture 2.1 (see Theorem 5.3).
\end{remark}

\section{Reflexivity, the PAMI and the MEP}

\setlength{\parindent}{2em}
Throughout this section, we let $\mathbb{F}$ be a finite field, $\Omega$ be a nonempty finite set, and $(k_{i}\mid i\in\Omega)$ be a family of positive integers. We consider the $\mathbb{F}$-vector space $\mathbf{H}\triangleq\prod_{i\in\Omega}\mathbb{F}^{k_{i}}$. Define the inner product $\langle~,~\rangle:\mathbf{H}\times\mathbf{H}\longrightarrow \mathbb{F}$ as $\langle\alpha,\beta\rangle=\sum_{i\in\Omega}\sum_{t=1}^{k_{i}}\alpha_{i,t}\cdot\beta_{i,t}$, where for $\alpha\in\mathbf{H}$ and $i\in\Omega$, $\alpha_{i,t}$ denote the $t$-th entry of $\alpha_{i}\in\mathbb{F}^{k_{i}}$. For any linear code (i.e., $\mathbb{F}$-subspace) $C\subseteq \mathbf{H}$, we let $C^{\bot}\triangleq\{\beta\in\mathbf{H}\mid \text{$\langle\alpha,\beta\rangle=0$ for all $\alpha\in C$}\}$ denote the dual code of $C$.

We also fix a non-trivial additive character $\chi$ of $\mathbb{F}$, and define the non-degenerate pairing $f:\mathbf{H}\times \mathbf{H}\longrightarrow\mathbb{C}^{*}$ as $f(\alpha,\beta)=\chi\left(\langle\alpha,\beta\rangle\right)$. It is well known that for any linear code $C\subseteq\mathbf{H}$, we have
\begin{equation}C^{\bot}=\{\beta\in\mathbf{H}\mid \text{$f(\alpha,\beta)=1$ for all $\alpha\in C$}\}.\end{equation}
For a partition $\Delta$ of $\mathbf{H}$, we let $\textbf{\textit{l}}(\Delta)$ denote the left dual partition of $\Delta$ with respect to $f$, and let $\inv(\Delta)$ denote the following subgroup of $\Aut_{\mathbb{F}}(\mathbf{H})$:
\begin{equation}\inv(\Delta)=\{\sigma\in\Aut_{\mathbb{F}}(\mathbf{H})\mid\text{$\beta\sim_{\Delta}\sigma(\beta)$ for all $\beta\in\mathbf{H}$}\}.\end{equation}
For a subgroup $K\leqslant\Aut_{\mathbb{F}}(\mathbf{H})$, we let $\orb(K)$ denote the \textit{orbit partition} of $K$ acting on $\mathbf{H}$, i.e., for any $\alpha,\beta\in\mathbf{H}$, $\alpha\sim_{\orb(K)}\beta$ if and only if there exists $\sigma\in K$ such that $\beta=\sigma(\alpha)$.

\setlength{\parindent}{0em}
\begin{definition}
{\bf{(1)}}\,\,Let $\Gamma$ and $\Lambda$ be partitions of $\mathbf{H}$. We say that $(\Lambda,\Gamma)$ admits MacWilliams identity if for any linear codes $C_1,C_2\subseteq\mathbf{H}$ such that $C_1\approx_{\Lambda}C_2$, it holds that ${C_1}^{\bot}\approx_{\Gamma}{C_2}^{\bot}$.

{\bf{(2)}}\,\,Let $\Delta$ be a partition of $\mathbf{H}$. We say that $\Delta$ satisfies the MacWilliams extension property (MEP) if for any linear code $C\subseteq\mathbf{H}$ and $g\in\Hom_{\mathbb{F}}(C,\mathbf{H})$ such that $g$ is injective and $\alpha\sim_{\Delta}g(\alpha)$ for all $\alpha\in C$, there exists $\varphi\in\inv(\Delta)$ such that $\varphi\mid_{C}=g$.

{\bf{(3)}}\,\,A partition $\Delta$ of $\mathbf{H}$ is said to be $\mathbb{F}$-invariant if for any $B\in\Delta$ and $c\in\mathbb{F}-\{0\}$, it holds that $B=\{c\cdot\beta\mid \beta\in B\}$.
\end{definition}

\setlength{\parindent}{2em}
We first examine the relations between reflexivity and the PAMI. The following lemma is an immediate consequence of Lemma 2.1, (2.6) and (5.1).

\setlength{\parindent}{0em}
\begin{lemma}
Let $\Gamma$ and $\Lambda$ be partitions of $\mathbf{H}$ such that $\Lambda$ is finer than $\textbf{\textit{l}}(\Gamma)$. Then, we have $\{0\}\in\Lambda$ and $(\Lambda,\Gamma)$ admits MacWilliams identity. Furthermore, for a reflexive partition $\Delta$ of $\mathbf{H}$, we have $\{0\}\in\Delta$, and both $(\textbf{\textit{l}}(\Delta),\Delta)$ and $(\Delta,\textbf{\textit{l}}(\Delta))$ admit MacWilliams identity.
\end{lemma}

\setlength{\parindent}{2em}
Now we show that the converse of Lemma 5.1 holds true for $\mathbb{F}$-invariant partitions. We begin with some basic properties. By [21, Remark 1.4], for an $\mathbb{F}$-invariant partition $\Theta$ of $\mathbf{H}$, $\textbf{\textit{l}}(\Theta)$ is again $\mathbb{F}$-invariant and is independent of the choice of the non-trivial additive character $\chi$, and the left generalized Krawtchouk matrix of $(\textbf{\textit{l}}(\Theta),\Theta)$ is independent of the choice of $\chi$ as well.

The following is our first main result of this section.

\setlength{\parindent}{0em}
\begin{theorem}
{Let $\Gamma$ and $\Lambda$ be $\mathbb{F}$-invariant partitions of $\mathbf{H}$. Then, the following three statements are equivalent to each other:

{\bf{(1)}}\,\,$\Lambda$ is finer than $\textbf{\textit{l}}(\Gamma)$;

{\bf{(2)}}\,\,$\{0\}\in\Lambda$, and $(\Lambda,\Gamma)$ admits MacWilliams identity;

{\bf{(3)}}\,\,$\{0\}\in\Lambda$, and for any $1$-dimensional linear codes $C_1,C_2\subseteq \mathbf{H}$ such that $C_1\approx_{\Lambda}C_2$, it holds that ${C_1}^{\bot}\approx_{\Gamma}{C_2}^{\bot}$.

Further assume that $|\Lambda|\leqslant|\Gamma|$. Then, (2) holds true if and only if $\Gamma$ is reflexive and $\Lambda=\textbf{\textit{l}}(\Gamma)$.
}
\end{theorem}

\begin{proof}
We note that $(1)\Longrightarrow(2)$ follows from Lemma 5.1 and $(2)\Longrightarrow(3)$ is trivial. Now we prove $(3)\Longrightarrow(1)$. Let $\textbf{S}$ denote the set of all the $1$-dimensional linear codes. Then, $\textbf{S}$ is a collection of non-identity subgroups of $\mathbf{H}$ with equal cardinality. Let $\Delta=\{C-\{0\}\mid C\in \textbf{S}\}\cup\{\{0\}\}$. Then, $\Delta$ is a partition of $\mathbf{H}$, and for any $A\in\Delta$ with $A\neq\{0\}$, there exists $C\in\textbf{S}$ such that $C-\{0\}=A$. Since $\Gamma$ is $\mathbb{F}$-invariant, $\textbf{\textit{l}}(\Gamma)$ is again $\mathbb{F}$-invariant and hence $\Delta$ is finer than $\textbf{\textit{l}}(\Gamma)$. Moreover, it follows from $\Lambda$ is $\mathbb{F}$-invariant that $\Delta$ is finer than $\Lambda$. Now we apply Theorem 2.1 and (5.1) and reach the fact that $\Lambda$ is finer than $\textbf{\textit{l}}(\Gamma)$, which further establishes $(3)\Longrightarrow(1)$. The rest is a direct consequence of Lemma 2.1 and the proven part $(1)\Longleftrightarrow(2)$.
\end{proof}

\setlength{\parindent}{2em}
Now we apply Theorem 5.1 to MacWilliams-type equivalence relations proposed in \cite{8} as well as partitions induced by weighted poset metric and combinatorial metric, as detailed in the following three examples.

\begin{example}$($Characterizing MacWilliams-type equivalence relations$)$
{Fix a poset $\mathbf{P}=(\Omega,\preccurlyeq_{\mathbf{P}})$, and consider an equivalence relation $E$ on $\mathcal{I}(\mathbf{P})$. Let $\Gamma$ denote the partition of $\mathbf{H}$ such that for any $\beta,\theta\in\mathbf{H}$, $\beta\sim_{\Gamma}\theta\Longleftrightarrow(\langle\supp(\beta)\rangle_{\mathbf{P}},\langle\supp(\theta)\rangle_{\mathbf{P}})\in E$, and let $\Lambda$ denote the partition of $\mathbf{H}$ such that for any $\alpha,\gamma\in\mathbf{H}$, $\alpha\sim_{\Lambda}\gamma\Longleftrightarrow(\Omega-\langle\supp(\alpha)\rangle_{\mathbf{\overline{P}}},\Omega-\langle\supp(\gamma)\rangle_{\mathbf{\overline{P}}})\in E$. Since $\Gamma$ and $\Lambda$ are $\mathbb{F}$-invariant partitions with $|\Lambda|=|\Gamma|$, Lemma 2.1 and Theorem 5.1 imply that $(\Lambda,\Gamma)$ is mutually dual with respect to $f$ if and only if $(\Lambda,\Gamma)$ admits MacWilliams identity and $\{\Omega\}$ is an equivalence class of $E$, which recovers $(i)\Longleftrightarrow(ii)$ of [8, Theorem 3.3]. In addition, the general MacWilliams identity (2.5) along with (5.1) recovers $(ii)\Longleftrightarrow(iii)$ of [8, Theorem 3.3]. In \cite{8}, $E$ is referred to as a MacWilliams-type equivalence relation if $(\Lambda,\Gamma)$ admits MacWilliams identity.
}
\end{example}

\begin{example}
{Let $\mathbf{P}=(\Omega,\preccurlyeq_{\mathbf{P}})$ be a poset, and fix $\omega:\Omega\longrightarrow\mathbb{R}^{+}$. Since $\mathcal{Q}(\mathbf{H},\mathbf{\overline{P}},\omega)$ and $\mathcal{Q}(\mathbf{H},\mathbf{P},\omega)$ are $\mathbb{F}$-invariant partitions of the same cardinality and containing $\{0\}$, Theorem 5.1 implies that $(\mathcal{Q}(\mathbf{H},\mathbf{\overline{P}},\omega),\mathcal{Q}(\mathbf{H},\mathbf{P},\omega))$ admits MacWilliams identity if and only if $\mathcal{Q}(\mathbf{H},\mathbf{\overline{P}},\omega)=\textbf{\textit{l}}(\mathcal{Q}(\mathbf{H},\mathbf{P},\omega))$. Further assume that $\mathbf{P}$ is hierarchical and $\omega$ is integer-valued. Then, by Theorem 3.2, $(\mathcal{Q}(\mathbf{H},\mathbf{\overline{P}},\omega),\mathcal{Q}(\mathbf{H},\mathbf{P},\omega))$ admits MacWilliams identity if and only if $\mathcal{Q}(\mathbf{H},\mathbf{P},\omega)$ is reflexive, if and only if $(\mathbf{P},\omega)$ satisfies the UDP, and for any $u,v\in\Omega$ such that $\len_{\mathbf{P}}(u)=\len_{\mathbf{P}}(v)$ and $\omega(u)=\omega(v)$, it holds that $k_{u}=k_{v}$. The latter equivalence has also been established in [29, Theorem 7] for labelled-poset-block metric by using different methods.
}
\end{example}

\setlength{\parindent}{0em}
\begin{example}
{Let $T$ be a covering of $\Omega$ such that $(T,\subseteq)$ is an anti-chain. Since $\mathcal{CO}(\mathbf{H},T)$ is an $\mathbb{F}$-invariant partition containing $\{0\}$, Theorems 4.1 and 5.1 imply that the following three statements are equivalent to each other:

{\bf{(1)}}\,\,$(\mathcal{CO}(\mathbf{H},T),\mathcal{CO}(\mathbf{H},T))$ admits MacWilliams identity;

{\bf{(2)}}\,\,$\mathcal{CO}(\mathbf{H},T)=\textbf{\textit{l}}(\mathcal{CO}(\mathbf{H},T))$;

{\bf{(3)}}\,\,$T$ is a partition of $\Omega$, and $\sum_{i\in U}k_{i}=\sum_{j\in V}k_{j}$ for all $U,V\in T$.

In addition, $(1)\Longleftrightarrow(3)$ recovers [39, Theorem 1] if $k_{i}=1$ for all $i\in\Omega$.
}
\end{example}

\setlength{\parindent}{2em}
Now we summarize the relations among reflexivity, the PAMI and the MEP in the following theorem.

\setlength{\parindent}{0em}
\begin{theorem}
Let $\Gamma$ be an $\mathbb{F}$-invariant partition of $\mathbf{H}$ such that $\{0\}\in\Gamma$, and consider the following five statements:

{\bf{(1)}}\,\,$\Gamma$ satisfies the MEP;

{\bf{(2)}}\,\,$\Gamma=\orb(\inv(\Gamma))$;

{\bf{(3)}}\,\,$\Gamma$ is reflexive;

{\bf{(4)}}\,\,$(\Gamma,\textbf{\textit{l}}(\Gamma))$ admits MacWilliams identity;

{\bf{(5)}}\,\,There exists an $\mathbb{F}$-invariant partition $\Lambda$ of $\mathbf{H}$ such that $\{0\}\in\Lambda$ and both $(\Lambda,\Gamma)$ and $(\Gamma,\Lambda)$ admit MacWilliams identity.

Then, it holds true that $(1)\Longrightarrow(2)$, $(2)\Longrightarrow(3)$ and $(3)\Longleftrightarrow(4)\Longleftrightarrow(5)$.
\end{theorem}

\begin{proof}
We begin by noting that $\{0\}\in\textbf{\textit{l}}(\Gamma)$ and $\textbf{\textit{l}}(\Gamma)$ is $\mathbb{F}$-invariant. Now $(2)\Longrightarrow(3)$ follows from [18, Theorem 2.6] and $(3)\Longrightarrow((4)\wedge(5))$ follows from Lemma 5.1. Next, we prove $(4)\Longrightarrow(3)$ and $(5)\Longrightarrow(3)$. If (4) holds true, then Theorem 5.1 implies that $\Gamma$ is finer than $\textbf{\textit{l}}(\textbf{\textit{l}}(\Gamma))$, which, along with Lemma 2.1, further implies that $\Gamma$ is reflexive, as desired. If (5) holds true, then Theorem 5.1 implies that $(\Lambda,\Gamma)$ is mutually dual with respect to $f$, which, along with Lemma 2.1, further implies that $\Gamma$ is reflexive, as desired. Hence it remains to establish $(1)\Longrightarrow(2)$. From now on, we assume that $\Gamma$ satisfies the MEP. By definition, it can be readily verified that $\orb(\inv(\Gamma))$ is finer than $\Gamma$. Next, we will show that $\Gamma$ is finer than $\orb(\inv(\Gamma))$, which immediately implies that $\Gamma=\orb(\inv(\Gamma))$. Consider $\alpha,\beta\in \mathbf{H}$ such that $\alpha\sim_{\Gamma}\beta$. If $\alpha=0$, then by $\{0\}\in\Gamma$, we have $\beta=\alpha=0$, and hence $\alpha\sim_{\orb(\inv(\Gamma))}\beta$, as desired. Therefore in what follows, we assume that $\alpha\neq0$. Then, by $\{0\}\in\Gamma$, we have $\beta\neq0$. Define $g\in\Hom_{\mathbb{F}}(\mathbb{F}\cdot \alpha,\mathbf{H})$ as $g(\alpha)=\beta$. Apparently, $g$ is injective. Moreover, it follows from $\Gamma$ is $\mathbb{F}$-invariant that $\gamma\sim_{\Gamma}g(\gamma)$ for all $\gamma\in\mathbb{F}\cdot \alpha$. Since $\Gamma$ satisfies the MEP, we can choose $\varphi\in\inv(\Gamma)$ such that $\varphi\mid_{\mathbb{F}\cdot \alpha}=g$. It follows that $\varphi(\alpha)=g(\alpha)=\beta$, which immediately implies that $\alpha\sim_{\orb(\inv(\Gamma))}\beta$, as desired. By now, we have shown that $\Gamma=\orb(\inv(\Gamma))$, which further establishes (2).
\end{proof}

\setlength{\parindent}{0em}
\begin{remark}
In Theorem 5.2, $(2)\Longrightarrow(1)$ does not hold in general. We refer the reader to [1, Example 2.9], [17, Example 1.12] and [21, Section 8] for counter-examples arising from partitions of matrix spaces and rank metric codes.
\end{remark}

\setlength{\parindent}{2em}
Finally, we disprove Conjecture 2.1. The following theorem is a direct consequence of Proposition 4.3, Theorem 4.2 and Theorem 5.2, in which we collect all the counter-examples to Conjecture 2.1 that we have obtained.

\setlength{\parindent}{0em}
\begin{theorem}
{{\bf{(1)}}\,\,Suppose that $|\mathbb{F}|=2$. Then, for any $n\in\mathbb{Z}^{+}$, $n\geqslant5$, $\mathcal{CO}(\mathbb{F}^{n},\mathcal{P}(3,[1,n]))$ does not satisfy the MEP. Moreover, fix $k\in\mathbb{Z}^{+}$, $k\geqslant4$. Then, there exists $m\in\mathbb{Z}^{+}$ such that for any $n\in\mathbb{Z}^{+}$ with $n\geqslant m$, $n\geqslant k+1$, $\mathcal{CO}(\mathbb{F}^{n},\mathcal{P}(k,[1,n]))$ does not satisfy the MEP. In addition, for any $n\in[k+2,2k]$, $\mathcal{CO}(\mathbb{F}^{n},\mathcal{P}(k,[1,n]))$ does not satisfy the MEP. Further assume that $2\nmid k$. Then, for any $n\in[\max\{k+2,7\},5k]$, $\mathcal{CO}(\mathbb{F}^{n},\mathcal{P}(k,[1,n]))$ does not satisfy the MEP.

{\bf{(2)}}\,\,Suppose that $|\mathbb{F}|\geqslant3$. Then, for any $n\in\mathbb{Z}^{+}$, $n\geqslant 3$, $\mathcal{CO}(\mathbb{F}^{n},\mathcal{P}(2,[1,n]))$ does not satisfy the MEP. Moreover, fix $k\in\mathbb{Z}^{+}$, $k\geqslant3$. Then, there exists $m\in\mathbb{Z}^{+}$ such that for any $n\in\mathbb{Z}^{+}$ with $n\geqslant m$, $n\geqslant k+1$, $\mathcal{CO}(\mathbb{F}^{n},\mathcal{P}(k,[1,n]))$ does not satisfy the MEP. In addition, $\mathcal{CO}(\mathbb{F}^{k+2},\mathcal{P}(k,[1,k+2]))$ does not satisfy the MEP, and for any $n\in\mathbb{Z}^{+}$ such that $n\geqslant k+1$, $n\equiv1~(\bmod~k)$, $\mathcal{CO}(\mathbb{F}^{n},\mathcal{P}(k,[1,n]))$ does not satisfy the MEP.
}
\end{theorem}

\begin{remark}
Fix $k\in\mathbb{Z}^{+}$, $k\geqslant3$, and consider the set
$$Q\triangleq\{n\in\mathbb{Z}^{+}\mid \text{$n\geqslant k+1$, $\mathcal{CO}({\mathbb{F}_2}^{n},\mathcal{P}(k,[1,n]))$ satisfies the MEP}\}.$$
By (1) of Theorem 5.3, $Q$ is a finite set. Moreover, if $k=3$, then $Q\subseteq\{4\}$; if $k\geqslant4$, then $Q\cap[k+2,2k]=\emptyset$; and if $2\nmid k$, then $Q\cap[\max\{k+2,7\},5k]=\emptyset$. In addition, (2) of Theorem 5.3 shows that similar conditions remain valid if $\mathbb{F}_2$ is replaced by a non-binary finite field.
\end{remark}

\section*{Appendix}\appendix

\section{Proof of Proposition A.1}

\setlength{\parindent}{2em}
In this appendix, we prove the following Proposition A.1, which has been used in the proof of Proposition 3.3.

\setlength{\parindent}{0em}
\begin{proposition}
{\bf{(1)}}\,\,Let $I$, $J$ be finite sets, $(n_{i}\mid i\in I)$, $(m_{j}\mid j\in J)$ be two families of positive integers, and $(a_{i}\mid i\in I)$, $(b_{j}\mid j\in J)$ be two families of positive real numbers. Assume that one of the following two equations holds:
\begin{equation}\footnotesize\prod_{i\in I}(x^{n_{i}}+a_{i})=\prod_{j\in J}(x^{m_{j}}+b_{j});\end{equation}
\begin{equation}\footnotesize\left(\prod_{i\in I}(x^{n_{i}}-a_{i})\right)\left(\prod_{j\in J}(x^{m_{j}}+b_{j})\right)=\left(\prod_{i\in I}(x^{n_{i}}+a_{i})\right)\left(\prod_{j\in J}(x^{m_{j}}-b_{j})\right).\end{equation}
Then, there exists a bijection $\sigma:I\longrightarrow J$ such that for any $i\in I$, $n_{i}=m_{\sigma(i)}$, $a_{i}=b_{\sigma(i)}$.

{\bf{(2)}}\,\,Let $Y$ be a finite set, $(n_{i}\mid i\in Y)$ be a family of positive integers, and $(a_{i}\mid i\in Y)$ be a family of positive real numbers. Fix $C,D\subseteq Y$ such that
$$\left(\prod_{i\in D}(x^{n_{i}}-1)\right)\left(\prod_{i\in Y-D}(x^{n_{i}}+a_{i})\right)=\left(\prod_{i\in C}(x^{n_{i}}-1)\right)\left(\prod_{i\in Y-C}(x^{n_{i}}+a_{i})\right).$$
Then, there exists a bijection $\sigma:C\longrightarrow D$ such that for any $i\in C$, $n_{i}=n_{\sigma(i)}$, $a_{i}=a_{\sigma(i)}$.
\end{proposition}

\begin{proof}
{\bf{(1)}}\,\,Without loss of generality, we assume that there exists $\gamma\in I$ such that for any $i\in I$ and $j\in J$, $n_{\gamma}\geqslant n_{i}$, $n_{\gamma}\geqslant m_{j}$. Let $\lambda\in\mathbb{C}$ be a $2n_{\gamma}$-th primitive root, and set $c=(a_{\gamma})^{1/n_{\gamma}}$. Then, $c\lambda$ is a root of $x^{n_{\gamma}}+a_{\gamma}$. For any $i\in I$, the facts $n_{\gamma}\geqslant n_{i}$, $a_{i}>0$ imply that $c\lambda$ is not a root of $x^{n_{i}}-a_{i}$. Therefore either (A.1) or (A.2) implies that there exists $\theta\in J$ such that $c\lambda$ is a root of $x^{m_{\theta}}+b_{\theta}$. By $n_{\gamma}\geqslant m_{\theta}$ and $b_{\theta}>0$, we have $m_{\theta}=n_{\gamma}$, $b_{\theta}=a_{\gamma}$. We then deduce that either (A.1) or (A.2) remains valid if we let $I-\{\gamma\}$ and $J-\{\theta\}$ take the place of $I$ and $J$, respectively. Applying an induction argument to $I-\{\gamma\}$ and $J-\{\theta\}$, and noticing that $n_{\gamma}=m_{\theta}$, $a_{\gamma}=b_{\theta}$, we conclude that there exists a bijection $\sigma:I\longrightarrow J$ such that $n_{i}=m_{\sigma(i)}$, $a_{i}=b_{\sigma(i)}$ for all $i\in I$, as desired.

{\bf{(2)}}\,\,Since for any $i,j\in Y$ with $a_{i}\neq1$, $\gcd(x^{n_{i}}+a_{i},x^{n_{j}}\pm1)=1$, we have the following two equations:
\begin{equation}
\hspace*{-1mm}\footnotesize\left(\prod_{i\in D}(x^{n_{i}}-1)\right)\left(\prod_{(i\in C,a_{i}=1)}(x^{n_{i}}+1)\right)=\left(\prod_{i\in C}(x^{n_{i}}-1)\right)\left(\prod_{(i\in D,a_{i}=1)}(x^{n_{i}}+1)\right),
\end{equation}
\begin{equation}\footnotesize\prod_{(i\in C,a_{i}\neq1)}(x^{n_{i}}+a_{i})=\prod_{(i\in D,a_{i}\neq1)}(x^{n_{i}}+a_{i}).\end{equation}
By (1) and (A.4), we can choose a bijection $\tau$ from $\{i\in C\mid a_{i}\neq1\}$ to $\{i\in D\mid a_{i}\neq1\}$ such that $n_{i}=n_{\tau(i)}$, $a_{i}=a_{\tau(i)}$ for all $i\in C$ with $a_{i}\neq1$. It follows that $\prod_{(i\in C,a_{i}\neq1)}(x^{n_{i}}-1)=\prod_{(i\in D,a_{i}\neq1)}(x^{n_{i}}-1)$, which, together with (A.3), further implies that
{\footnotesize\begin{eqnarray*}\left(\prod_{(i\in D,a_{i}=1)}(x^{n_{i}}-1)\right)\left(\prod_{(i\in C,a_{i}=1)}(x^{n_{i}}+1)\right)=\left(\prod_{(i\in C,a_{i}=1)}(x^{n_{i}}-1)\right)\left(\prod_{(i\in D,a_{i}=1)}(x^{n_{i}}+1)\right).\end{eqnarray*}}By (1), we can choose a bijection $\eta$ from $\{i\in C\mid a_{i}=1\}$ to $\{i\in D\mid a_{i}=1\}$ such that $n_{i}=n_{\eta(i)}$ for all $i\in C$ with $a_{i}=1$. Define $\varepsilon:C\longrightarrow D$ as $\varepsilon\mid_{\{i\in C\mid a_{i}\neq1\}}=\tau$ and $\varepsilon\mid_{\{i\in C\mid a_{i}=1\}}=\eta$. It follows that $\varepsilon:C\longrightarrow D$ is a bijection such that $n_{i}=n_{\varepsilon(i)}$, $a_{i}=a_{\varepsilon(i)}$ for all $i\in C$, as desired.
\end{proof}

\end{document}